%% file: main.tex
\documentclass[10pt]{article}
\usepackage[a4paper,margin=1.2in]{geometry} 
\usepackage{amsmath,amssymb,amsfonts,amsthm} %
\usepackage{algorithm} %
\usepackage[titletoc,title]{appendix} %
\usepackage{authblk} %
\usepackage[english]{babel} %
\usepackage{bbm} %
\usepackage{booktabs} %
\usepackage{color} %
\usepackage{derivative} %
\usepackage{dsfont} %
\usepackage{enumitem} %
\usepackage{float} %
\usepackage[T1]{fontenc} %
\usepackage{lmodern} %
\usepackage{graphicx} %
\usepackage[colorlinks,citecolor=blue,bookmarks=true,breaklinks=true]{hyperref} %
\usepackage[capitalize]{cleveref} 
\usepackage{hyphenat} %
\usepackage[utf8]{inputenc} %
\usepackage{interval} %
\usepackage{mathtools} %
\usepackage{algpseudocode} %
\usepackage{microtype} %
\usepackage{physics} %
\usepackage{soul} %
\usepackage{subcaption} %
\usepackage{tikz} %
\usepackage{xparse} %
\usepackage{xspace} %
\usepackage{url}
\usepackage{xr}
\graphicspath{{images/}}

\setul{1ex}{.5pt} %

\setlist[enumerate]{nolistsep,itemsep=3pt,topsep=3pt,leftmargin=*} %
\setlist[itemize]{nolistsep,itemsep=3pt,topsep=3pt,leftmargin=2em} %

\intervalconfig{soft open fences}

\newfloat{algorithm}{hbt}{lop} %
\floatname{algorithm}{Algorithm} %
\crefname{algorithm}{Algorithm}{Algorithms} %

\crefname{subsection}{Subsection}{Subsections} %

\usetikzlibrary{cd,calc} %

\theoremstyle{plain} 
\newtheorem{theorem}{Theorem}
\newtheorem{lemma}{Lemma}
\newtheorem{Corollary}{Corollary}

\theoremstyle{definition}
\newtheorem{definition}[theorem]{Definition}

\newcommand{\microspace}{\mspace{.5mu}} %

\newcommand{\avg}[1]{\langle\microspace#1\microspace\rangle} %
\newcommand{\noise}[1]{\widetilde{#1}} %
\newcommand{\navg}[1]{\langle\microspace\noise{#1}\microspace\rangle} %
\newcommand{\bignavg}[1]{\bigl\langle\microspace\noise{#1}\microspace\bigr\rangle} %
\newcommand{\Bignavg}[1]{\Bigl\langle\microspace\noise{#1}\microspace\Bigr\rangle} %

\newcommand{\labelstyle}[1]{\textrm{#1}} %
\newcommand{\lin}{\labelstyle{in}} %

\newcommand{\combine}{\mathsf{combine}} 
\newcommand{\neighbor}{\mathsf{neighbor}} 
\newcommand{\vqe}{\mathsf{vqe}}

\DeclareMathOperator*{\argmin}{argmin} %
\def\H{\mathcal{H}} 
\DeclareMathOperator{\Exp}{\mathbb{E}} 
\DeclareMathOperator{\Var}{Var}

\def\C{\mathcal{C}} 
\def\E{\mathcal{E}} 
\def\P{\mathcal{P}} 
\def\Q{\mathcal{Q}} 
\def\U{\mathcal{U}} 
\def\V{\mathcal{V}} 
\def\K{\mathcal{K}} 
\def\X{\mathcal{X}} 
\def\Y{\mathcal{Y}} 
\def\Z{\mathcal{Z}} 
\def\S{\mathcal{S}} 

\def\CC{\mathbb{C}}

\def\GG{\mathbb{G}} 

\newcommand{\diag}{\operatorname{diag}}

\newcommand{\ExyS}{\underset{\left(\vb{x}, y\right) \sim P_{\mathbb{S}}}{\Exp}}
\newcommand{\ES}{\underset{P_{\mathbb{S}}}{\Exp}}
\newcommand{\Et}{\underset{P_{\mathbb{\rm test}}}{\Exp}}
\newcommand{\Exyt}{\underset{\left(\vb{x}, y\right) \sim P_{\rm test}}{\Exp}}
\newcommand{\Em}{\underset{\mathcal{M}}{\Exp}}
\newcommand{\Etrain}{\underset{\operatorname{train}}{\Exp}}
\newcommand{\Etest}{\underset{\operatorname{test}}{\Exp}}
\newcommand{\ReNs}{\vb{\Lambda}_{N_s}}

\begin{document}
\title{Scalable Quantum Error Mitigation with\\ Neighbor-Informed Learning}
\renewcommand\Affilfont{\itshape\small}%
\author[1]{Zhenyu Chen\thanks{These authors contributed equally to this work.}}
\author[2]{Bin Cheng$^\ast$}
\author[3,4]{Minbo Gao}
\author[5,6]{Xiaodie Lin}
\author[7,8]{Ruiqi Zhang}
\author[7,9]{Zhaohui Wei\thanks{Corresponding author. Email: weizhaohui@gmail.com}}
\author[1,11]{Zhengfeng Ji\thanks{Corresponding author. Email: jizhengfeng@tsinghua.edu.cn}}
\affil[1]{Department of Computer Science and Technology, Tsinghua University, Beijing, China.}
\affil[2]{Centre for Quantum Technologies, National University of Singapore, Singapore.}
\affil[3]{Institute of Software, Chinese Academy of Sciences, Beijing, China}
\affil[4]{University of Chinese Academy of Sciences, Beijing, China}
\affil[5]{Department of Mechanical and Automation Engineering, The Chinese University of Hong Kong, Shatin, Hong Kong SAR, China.}
\affil[6]{College of Computer and Data Science, Fuzhou University, Fuzhou, China.}
\affil[7]{Yau Mathematical Sciences Center, Tsinghua University, Beijing, China.}
\affil[8]{Department of Mathematics, Tsinghua University, Beijing, China.}
\affil[9]{Yanqi Lake Beijing Institute of Mathematical Sciences and Applications, Beijing, China.}
\affil[10]{Zhongguancun Laboratory, Beijing, China.}
\date{}

\maketitle

\begin{abstract}
  Noise in quantum hardware is the primary obstacle to realizing the
  transformative potential of quantum computing.
  Quantum error mitigation (QEM) offers a promising pathway to enhance
  computational accuracy on near-term devices, yet existing methods face a
  difficult trade-off between performance, resource overhead, and theoretical
  guarantees.
  In this work, we introduce neighbor-informed learning (NIL), a versatile and
  scalable QEM framework that unifies and strengthens existing methods such as
  zero-noise extrapolation (ZNE) and probabilistic error cancellation (PEC),
  while offering improved flexibility, accuracy, efficiency, and robustness.

  NIL learns to predict the ideal output of a target quantum circuit from the
  noisy outputs of its structurally related ``neighbor'' circuits.
  A key innovation is our \emph{2-design training} method, which generates
  training data for our machine learning model.
  In contrast to conventional learning-based QEM protocols that create training
  circuits by replacing non-Clifford gates with uniformly random Clifford gates,
  our approach achieves higher accuracy and efficiency, as demonstrated by both
  theoretical analysis and numerical simulation.
  Furthermore, we prove that the required size of the training set scales only
  \emph{logarithmically} with the total number of neighbor circuits, enabling
  NIL to be applied to problems involving large-scale quantum circuits.
  Our work establishes a theoretically grounded and practically efficient
  framework for QEM, paving a viable path toward achieving quantum advantage on
  noisy hardware.
\end{abstract}

\section{Introduction}
Quantum computing is expected to deliver significant speedups across a variety
of computational problems, including quantum
simulation~\cite{Feynman1982-simulation, Lloyd1996-simulation}, integer
factorization~\cite{shor_algorithms_1994}, unstructured
search~\cite{grover_fast_1996}, and many other computational problems.
However, realizing this promise requires large-scale quantum error correction,
which imposes significant resource demands that currently lie far beyond
existing technological capabilities.
At present, quantum devices operate in the noisy intermediate-scale quantum
(NISQ) era, where the qubit count remains insufficient for quantum error
correction and the systems are significantly affected by
noise~\cite{preskill_quantum_2018}.
Consequently, it would be desirable to mitigate the impact of noise in quantum
information processing tasks with techniques that do not entail substantial
quantum overheads.
In recent years, a widely applied approach of this kind is the so-called
\emph{quantum error mitigation} (QEM)~\cite{temme_error_2017,
  li_efficient_2017,aharonov2025on}.

The two most prominent quantum error mitigation techniques are zero-noise
extrapolation (ZNE)~\cite{temme_error_2017, li_efficient_2017} and probabilistic
error cancellation (PEC)~\cite{temme_error_2017}.
Several additional quantum error mitigation methods have also been developed,
including symmetry verification~\cite{bonet-monroig_low-cost_2018,
  mcardle_error-mitigated_2019}, purification-based
approaches~\cite{huggins_virtual_2021, koczor_exponential_2021, liu2024virtual},
and machine learning based techniques~\cite{liao2024machine,czarnik_error_2021,
  strikis_learning-based_2021,lowe_unified_2021}.
Specifically, ZNE enhances the accuracy of quantum computation by artificially
amplifying the noise in a quantum circuit and then extrapolating the noisy
outputs back to the zero-noise limit.
This approach is appealing due to its low experimental overhead, and has been
experimentally demonstrated on superconducting devices~\cite{kandala_error_2019,
  kim_scalable_2023}.
However, its accuracy is ultimately constrained by the quality of
extrapolations, and offers little flexibility for further improvement with
additional computational resources.
The PEC method directly mitigates noise by inverting each noise channel through
a quasiprobability decomposition into physically implementable
operations~\cite{temme_error_2017}.
It can, in principle, yield unbiased estimates and has been shown to be
theoretically optimal in certain settings~\cite{takagi_optimal_2021,
  takagi_fundamental_2021, jiang2021physical}.
Yet, PEC suffers from a fundamental scalability barrier, i.e., it requires an
exponential number of circuit evaluations to suppress noise to high precision.
Moreover, it typically depends on accurate noise characterization, which is
often difficult to obtain in practice.
To address the overhead issue of PEC, heuristic strategies have been proposed,
such as truncating the inverse channel expansion and applying learning-based
methods to estimate mitigation
coefficients~\cite{strikis_learning-based_2021,van2023probabilistic}.
While promising, these approaches generally lack rigorous performance guarantees
and are typically benchmarked only on small-scale instances or specific classes
of quantum circuits.

In many practical applications of near-term quantum algorithms, achieving higher
accuracy is often a priority, and allocating additional computational resources
to improve precision can be well justified.
This motivates the following fundamental question: \textit{Is it possible to
  design a quantum error mitigation strategy that is flexible---allowing
  resource allocation to be dynamically adjusted based on the desired accuracy;
  efficient---circumventing the computational overhead inherent in methods such
  as probabilistic error cancellation (PEC); and theoretically sound---with
  provable performance guarantees under reasonable assumptions?}
In this work, we provide an affirmative answer to this question.
Specifically, we introduce a general and principled framework for quantum error
mitigation, termed neighbor-informed learning (NIL), and demonstrate its strong
potential to address these challenges effectively.
Much like a detective solving a case by interviewing witnesses rather than
confronting the suspect directly, NIL reconstructs the ideal output of a target
circuit by aggregating observations from its surrounding neighbor circuits,
which offer noisy and partial reflections of the original circuit.
Our method is illustrated in \cref{fig:neighborhood_learning} and described in
detail in \cref{sec:framework}.

\begin{figure*}[htbp!]
  \centering%
  \includegraphics[width = \textwidth]{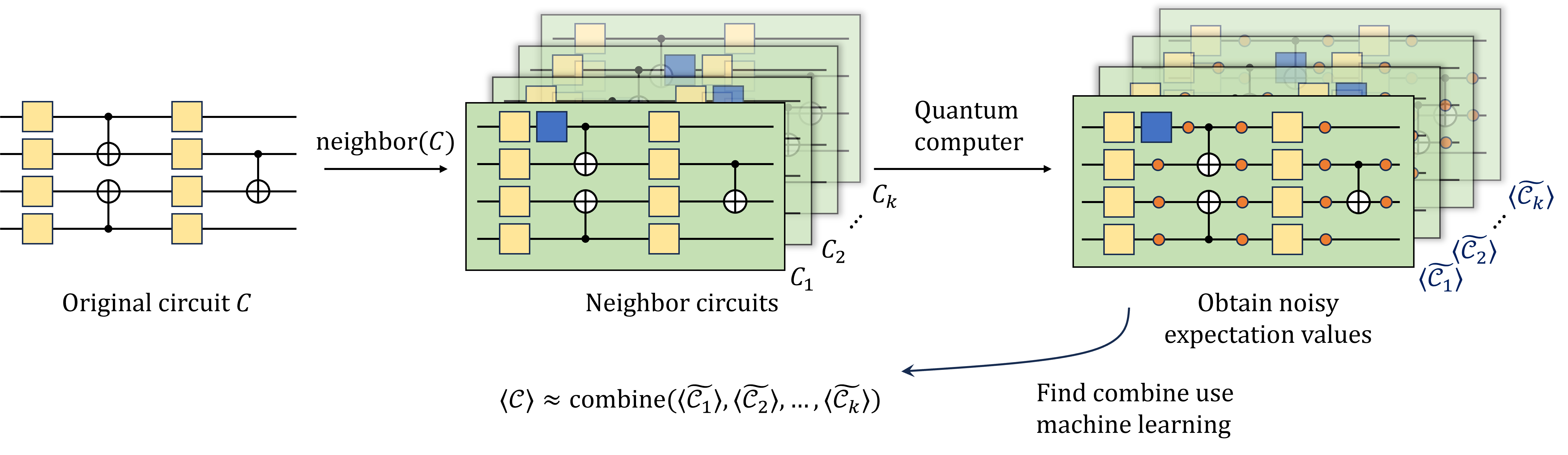}
  \caption{Schematic illustration of NIL\@.
    The single-qubit gates in the original circuit are represented by yellow
    blocks, the inserted gates are represented by blue blocks, and the orange
    circles represent the noise channels.
    The noisy expectation values $\navg{\mathcal{C}_j}$ are obtained on a real
    quantum device.
    The function $\combine$ is learned on a classical computer from a training
    set of Clifford circuits, for which the ideal expectation values can be
    efficiently calculated classically.}%
  \label{fig:neighborhood_learning}
\end{figure*}

This strategy provides substantial flexibility in the design of both the
neighbor circuits and the learning model, making it broadly compatible with
families of parameterized quantum circuits commonly used in NISQ algorithms such
as
VQE~\cite{peruzzo_variational_2014,yung_transistor_2014,McClean2016,Kandala2017}
and QAOA~\cite{farhi_quantum_2014,farhi2016quantum,moll2018quantum}.
In particular, we explore a variety of neighbor constructions including
gate-insertion schemes that yield Pauli and CPTP-basis neighbor circuits, as
well as those inspired by ZNE\@.
We find that combining different types of neighbor circuits can further enhance
performance while reducing computational cost.
Moreover, we would like to stress that by training our model once on a
representative set of training circuits, the learned model can be reused across
many circuits of the same configuration.

As a learning-based method, the performance of NIL depends heavily on the
quality of training datasets.
To this end, we propose a novel \emph{2-design training} strategy tailored for
quantum circuit architectures with rotational parameters of quantum gates.
The key idea is to replace each non-Clifford gate in the target circuit with a
random Clifford gate forming a rotational $2$-design, rather than using
uniformly random Clifford gates as suggested by the random Clifford method
adopted in previous works~\cite{strikis_learning-based_2021}.
We show that the data collected from these training circuits can be used to
train an optimal map $\combine$ selected among the class of linear functions in
an average-case sense.
In particular, this 2-design training method demonstrates strong generalization
ability when linear regression or Lasso regression is used to learn the
$\combine$ map.

Furthermore, we prove that to compute the target expectation value within an
error $\varepsilon$, the size of the training set required scales as
$\mathcal{O}(\ln{N}/\varepsilon^2)$, where $N$ is the number of neighboring
circuits.
This enables the method to leverage a large number of neighbor circuits, which
typically enhances the quality of the estimation and improves both efficiency
and flexibility.
Notably, our bound on the training set size also incorporates the effect of
measurement shot noise.
We rigorously analyze the error introduced by shot noise and show that the
number of measurement shots required is only $\mathcal{O}(1/\varepsilon^2)$.
Additionally, we demonstrate that shot noise can have a beneficial regularizing
effect on the method.

Through extensive numerical simulations, we demonstrate that, compared to
various existing QEM techniques, our method achieves a superior tradeoff between
error mitigation performance and computational resource consumption.
For example, when applied to the VQE circuits for the LiH molecule, our method
achieves a four-order-of-magnitude improvement in MSE at the same training cost,
compared to the random Clifford training
method~\cite{czarnik_error_2021,strikis_learning-based_2021,lowe_unified_2021}.
Moreover, we validate the scalability of our method on circuits with over 100
qubits and more than 20 layers, demonstrating that decent QEM performance is
still achievable even at scale.
Several other technical design choices, such as the construction of neighbor
circuits and the selection of learning models and algorithms, are also
considered numerically.
Additional details can be found in the Appendix.


\section{Neighbor-Informed Learning}\label{sec:framework}

\paragraph{General framework.}
\Cref{fig:neighborhood_learning} illustrates the workflow of NIL\@.
As a general and flexible framework for QEM, NIL reconstructs ideal expectation
values by leveraging noisy outcomes from structurally related circuits, referred
to as \emph{neighbor circuits}.
Instead of relying solely on the noisy outcome of a target circuit
$\mathcal{C}$, NIL collects noisy outcomes from a set of modified circuits
$\{\mathcal{C}_1, \ldots, \mathcal{C}_N\}$ and uses a learned function
$\combine$ to approximate the ideal result:
\begin{equation*}
  \avg{\mathcal{C}} \approx \combine \bigl( \navg{\mathcal{C}_1},
    \ldots, \navg{\mathcal{C}_{N}} \bigr).
\end{equation*}
Here, $\avg{\mathcal{C}}$ denotes the ideal, noise-free expectation value of the
target circuit $\mathcal{C}$, while $\navg{\mathcal{C}_j}$ represents the noisy
expectation value obtained from the $j$-th neighbor circuit on a quantum device.
This formulation generalizes and unifies several existing QEM approaches.
For example, in ZNE, neighbor circuits are created by artificially amplifying
noise, and $\combine$ corresponds to an extrapolation
function~\cite{temme_error_2017, endo_practical_2018}.
In PEC, neighbor circuits are generated by inserting physically implementable
operations into the original circuit according to quasiprobability distributions
that invert the noise channels, and $\combine$ is the corresponding weighted
sum~\cite{takagi_optimal_2021}.

To implement NIL in practice, we introduce several strategies to construct
neighbor circuits: (i) inserting single-qubit gates after noisy operations
(referred to as Pauli or CPTP neighbors), and (ii) generating ZNE-style
neighbors by scaling the noise rate (see \cref{sec:methods} for detailed constructions).
For a quantum circuit, the noisy outputs from its neighbor circuits serve as the
features of this circuit, which will be utilized in our machine learning models.

\paragraph{2-design training method.}
The function $\combine$ is determined using machine learning methods and
therefore requires a high-quality training dataset.
A central challenge in learning-based QEM is generating training data that
statistically mirrors the target circuits, which often feature continuous
rotational parameters.
To address this, we propose the \emph{2-design training
  method}, which faithfully respects this rotational parametrization structure.
This new methods
outperforms the conventional \emph{random Clifford} training method widely
used in learning-based
QEM~\cite{czarnik_error_2021,strikis_learning-based_2021,lowe_unified_2021}.
While the Clifford training method constructs training circuits by replacing
each non-Clifford single-qubit gate with a uniformly random Clifford gate, our
method replaces each rotation gate $R_P(\theta)$ with a Clifford gate uniformly
drawn from the four-element set $\{R_P(0), R_P(\pi/2), R_P(\pi), R_P(3\pi/2)\}$
where $P\in\{X,Y,Z\}$ is a Pauli operator.
The selection of these four Clifford gates is motivated by a key mathematical
property that we observe: they constitute a quantum \emph{rotational} 2-design that
captures the second-order statistical properties of Pauli rotations
$R_P(\theta)$, with $\theta$ uniformly distributed over $[0, 2\pi]$.
A formal definition and mathematical characterizations of this 2-design
construction are provided in Appendix~B.

In this work, we assume that target circuits consist of Pauli rotation gates and
Clifford gates, which together form a universal set.
The training circuits generated by the 2-design method remain Clifford circuits.
For each training circuit, we generate its noisy neighbor circuits and collect
their outputs as features.
The ideal expectation value of the training circuit, which is classically
simulable due to its Clifford
structure~\cite{gottesman_heisenberg_1999,aaronson_improved_2004}, serves as the
label.

\begin{figure*}[t!]
  \centering
  \begin{subfigure}[t]{0.7\textwidth}
      \includegraphics[width=\textwidth]{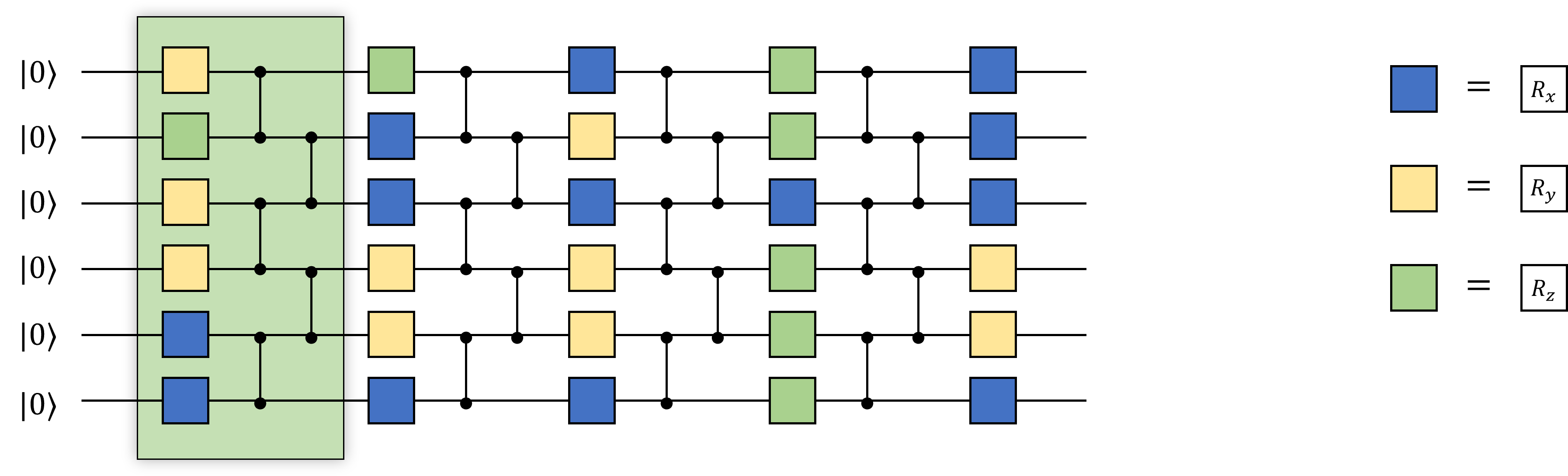}
      \caption{$\mathsf{vqe}$-6-4 (6 qubits, 13 layers)}
  \end{subfigure}
   
  \begin{subfigure}[t]{0.7\textwidth}
    \includegraphics[width=\textwidth]{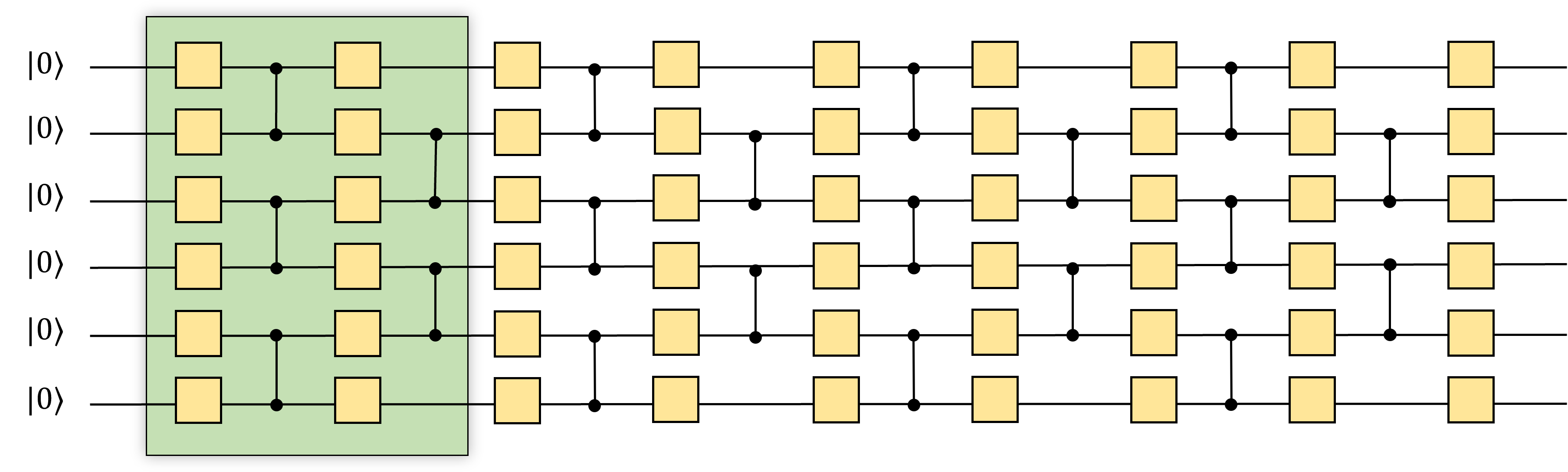}
    \caption{$\mathsf{vqe}$-$R_y$-$6$-$4$ (8 qubits, 17 layers)}
\end{subfigure}
  \vfill
  \begin{subfigure}[t]{0.7\textwidth}
    \includegraphics[width=\textwidth]{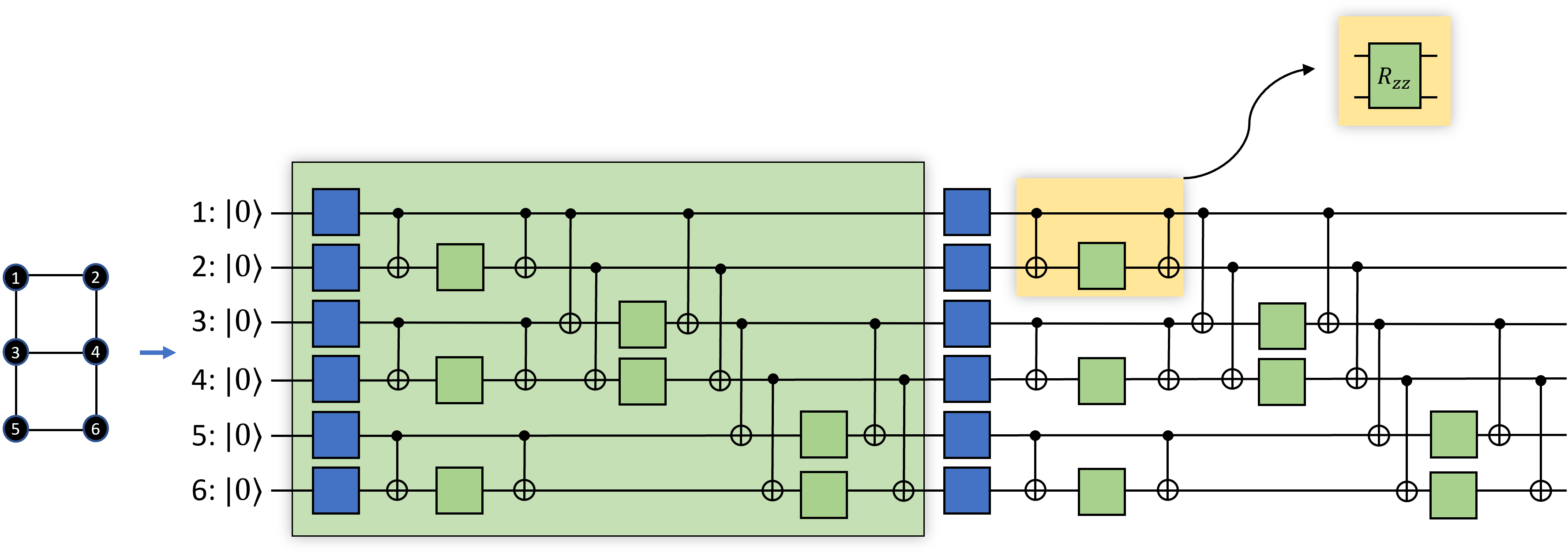}
    \caption{$\mathsf{vqe}$-$(3, 2)$-$2$ (6 qubits, 20 layers)}
\end{subfigure}
\caption{Three types of ansatz circuits we employ to test NIL\@.
  In these diagrams, colored boxes represent single-parameter rotation gates,
  vertical black lines are two-qubit Clifford gates, and the areas enclosed by
  dashed lines denote a complete block.
  All three circuits contain six qubits and can be extended to an arbitrary
  number of qubits.}\label{fig:ansatz_circuits}
\end{figure*} 

By repeating this procedure, we construct a sufficiently large training dataset.
We then apply linear models to fit the data and obtain a linear estimator of the
ideal expectation value for QEM, which will be ultimately applied to the target
quantum circuit.
Among various types of linear models, we find that \emph{Lasso regression}
offers the best trade-off between prediction accuracy and computational
efficiency.
The 2-design training method is described in detail in Appendix~A.

\section{Mathematical Characterizations of NIL}
\paragraph{Provable optimality.}
NIL offers several appealing practical advantages for QEM\@.
First, we prove that NIL always achieves the optimal performance across all
linear models, regardless of the presence of shot noise.
The proof relies on the observation that, when reformulating the linear
regression problem as a quadratic optimization problem, each entry of the
resulting coefficient matrix can be expressed as as the trace of a second-moment
expression involving the rotation gates $R_P(\theta)$ and $R_P(-\theta)$.
Notably, the statistical properties of these coefficients can be mimicked by
their behavior at the specific rotation angles
$\{R_P(0), R_P(\pi/2), R_P(\pi), R_P(3\pi/2)\}$.
This implies that in the large size limit, the training set matches the test set
on average, and the solution to the linear (or Lasso) regression exhibits
excellent generalization capabilities.

\begin{theorem}[informal]\label{thm:convergence_optimal}
  As the size of the training set approaches infinity, the linear model trained
  with the 2-design training method converges to the optimal linear estimator
  for the target circuits.
\end{theorem}

The detailed proof of \cref{thm:convergence_optimal} is provided in
Appendix~B.
We establish this result under both noiseless and shot-noise-affected
expectation value settings.

\paragraph{Limitation of the Clifford training method.}
Since single-qubit rotational gates with parameters are prevalent in many
near-term quantum algorithms, the 2-design training method offers broad
applicability and practical relevance.
Importantly, we prove that the $2$-design training method is provably superior
to the Clifford training method.
The former is not only more efficient and accurate, but also resolves a key
limitation of the latter: In the Clifford training method, the mean squared
error (MSE) obtained on the training dataset fails to converge to the
corresponding value on the target circuit in the average sense, and as a result,
the generalization error of the model is non-negligible, seriously limiting its
performance.

On the other hand, if we generate the training dataset with the 2-design
training method, the MSE of any linear model evaluated on the training dataset
exactly matches the MSE on the target parameterized quantum circuit.
This crucial equivalence ensures the advantage of our approach over the Clifford
training method that is widely applied in the literature.
Furthermore, we would like to remark that this property also allows us to
benchmark the performance of our approach on large-scale target quantum circuits
that are hard to simulate classically, as it only requires us to compute the
ideal expectation values for Clifford training circuits.
The proof of this property can be found in Appendix~D.

\paragraph{Scalability and efficiency of NIL.}
We now analyze the size of the training dataset required to obtain a reliable
linear estimator for practical deployment.
We prove that using Lasso regression with a constraint on the $\ell_1$-norm of
the linear estimator, our model requires only a \emph{logarithmic} number of
training circuits in $N$, the number of neighbor circuits.
This ensures that a good estimator can be learned efficiently for mitigating
noise in the target circuit, supporting scalability to large-scale quantum
circuits.
Additionally, we derive the classical runtime of NIL, showing its time
complexity remains polynomial in $N$.
The detailed derivations are provided in Appendix~C.
These properties collectively establish NIL as a scalable and practical approach
for QEM\@.

NIL possesses several additional noteworthy properties.
First, when examining the effect of shot noise caused by quantum measurements,
we find that its presence introduces an additional helpful regularization term
into the original linear regression problem from a statistical perspective.
This implies that a small amount of shot noise, which is common in practical
scenarios, can actually enhance the numerical stability of the linear
regression.
Moreover, we rigorously prove that to obtain a linear estimator with an additive
error $\varepsilon$ with NIL, the number of measurement shots required for
estimating each noisy neighbor circuit's expectation value is
$\mathcal{O}(1/\varepsilon^2)$.
Combining these facts, we have the following theorem.

\begin{theorem}[informal]\label{thm:convergence_MSE}
  When using the 2-design method to generate the training set in NIL, it
  suffices to choose the size of the training set $T$ to be
  $\mathcal{O}(\ln N / \varepsilon^{2})$ to achieve an error of $\varepsilon$ in
  MSE\@.
\end{theorem}

The detailed proofs are provided in Appendix~B and Supplementary
Appendix~C.

\section{Numerical Simulations and Comparison}\label{sub:test_set}

We now illustrate the effectiveness of NIL in QEM by benchmarking it against
existing QEM methods.
First, we provide numerical evidence supporting our theoretical result that
unlike the Clifford training method, NIL consistently converges to the optimal
linear estimator.
Second, motivated by the recent seminal experimental work~\cite{kim2023evidence}
demonstrating the power of the ZNE protocol in QEM, we show that its performance
can be further enhanced by integrating our approach.

We first conduct numerical experiments on several small-scale variational
quantum algorithms, for which the ideal expectation values of the circuits can be
classically simulated.
With access to the exact outputs of these non-Clifford circuits, we are able to
compare the performance of NIL with that of conventional QEM methods such as ZNE\@.
We examine three representative circuit types: variational ansatz circuits for
1D and 2D transverse-field Ising (TFI) models, and unitary coupled-cluster (UCC)
circuits for simulating the LiH molecule.
Unless stated otherwise, the noise model applied is local depolarizing noise with
a strength of 0.001 for single-qubit gates and 0.01 for two-qubit gates.
Importantly, our method does not require prior knowledge of the noise model and thus
remains applicable to arbitrary noise models in practical settings.

\paragraph{1D and 2D transverse-field Ising model.}
Variational Quantum Eigensolver (VQE) is a class of hybrid quantum-classical
algorithms designed to find the minimum eigenvalues of Hamiltonians, which has
been widely applied in quantum chemistry and
physics~\cite{peruzzo_variational_2014,yung_transistor_2014,McClean2016,Kandala2017}.
The basic workflow of VQE involves preparing a quantum state using a quantum
circuit with tunable parameters (called the ansatz circuit), measuring the
energy of the system, and then using a classical optimizer to iteratively adjust
the circuit parameters to minimize the energy.
Here, we use VQE as a testbed for our NIL technique.

Specifically, we consider two classes of target problems: the
transverse-field Ising (TFI) model and a quantum chemistry problem.
Given a graph $G = (V, E)$, the TFI Hamiltonian is defined as
\begin{align*}
  H_{\rm TFI} = -J \sum_{(i, j) \in E} Z_i Z_j - h \sum_{j \in V} X_j,
\end{align*}
where we set $J = 1$ and $h = 2$.
We consider two different graphs: a 1D line and a 2D grid.

For the 1D case, we choose two kinds of hardware-efficient
ansatz~\cite{Kandala2017, strikis_learning-based_2021}.
The first one, labeled as $\mathsf{vqe}$-$n$-$m$ and depicted in
\cref{fig:ansatz_circuits} (a), involves applying the block enclosed by the red
dotted line $m$ times, where $n$ is the number of qubits.
Inside a block, there is a layer of random Pauli rotations in the $X$, $Y$, or
$Z$ direction and a layer of CZ gates.
The second one, labeled as $\mathsf{vqe}$-$R_y$-$n$-$m$ and depicted in
\cref{fig:ansatz_circuits} (b), has a similar structure, but in this case there
are two layers of Pauli-$Y$ rotations and two layers of CZ gates inside a block.

For the 2D case, we choose the so-called Hamiltonian variational ansatz
(HVA)~\cite{wecker_progress_2015}.
For an $n_1 \times n_2$ grid, the ansatz circuit is labeled as
$\vqe$-$(n_1, n_2)$-$m$ and defined by
\begin{align*}
  \prod_{t=1}^m \exp(i \sum_{j \in V} \alpha_{j}^{(t)} X_j)
  \exp(i \sum_{(j,k)\in E} \beta_{jk}^{(t)} Z_j Z_k) \ ;
\end{align*}
see \cref{fig:ansatz_circuits} (c) for an example.
Note that unlike the original HVA, here we let all the angles in the rotation
gates be tunable.

On these problems we first compare the performance of NIL with that of the
learning-based QEM protocol introduced in
Ref.~\cite{strikis_learning-based_2021}.
For this, we adopt the neighbor generation strategy proposed therein: inserting
random Pauli operators at up to three randomly selected positions in both the
training and the target circuits.
We refer to the resulting circuits as \emph{weight-$3$ neighbors}, detailed
construction and discussion are provided in \cref{sec:methods}.
Based on this set of neighbor circuits, we prepare two training datasets of size
5,000 each: one generated via the 2-design training method, and the other via
the Clifford training method that replaces each non-Clifford single-qubit gate
by uniformly random Clifford gates.
Lasso regression is used to produce a linear estimator for QEM separately on
each dataset.

To evaluate and compare the performance of the obtained linear estimators for
QEM, we construct a test set of 1,000 non-Clifford circuits with rotation angles
uniformly sampled from $[0, 2\pi]$.
We generate corresponding neighbor circuits for these test circuits with the
same strategy as for the training circuits, then collect the noisy expectation
values for them as features and obtain the ideal expectation values as labels
via classical simulations.
Then we apply the two linear estimators obtained by 2-design training method and
the Clifford training method on the test dataset, and the corresponding MSEs can
be seen in \cref{fig:weight123_24_9_val}.
Note that to benchmark the performance of the two linear estimators, we also
apply Lasso regression directly to fit the test dataset, and obtain a third
linear estimator.

In \cref{fig:weight123_24_9_val} it can be seen that the MSE of the estimator
learned from NIL (blue curve) consistently decreases as the number of neighbor
circuits increases.
Moreover, the blue curves closely follow the linear estimator learned directly
from the test dataset (green curve), providing strong empirical support for our
theoretical claim that NIL can accurately approximate the optimal linear model.
In contrast, the Clifford training method (orange curve) consistently yields
higher test error.

Meanwhile, it can be observed that using only weight-$1$ neighbors---whose
number equals the number of gates in the target circuit---already yields
sufficiently low MSE, making it practical for QEM\@.
When restricting to weight-$1$ neighbors, NIL achieves approximately half the
MSE of the Clifford training method for VQE circuits for the 1D TFI Hamiltonian,
and an advantage of one order of magnitude for the 2D TFI Hamiltonian.
These results confirm once again that, for the target circuits considered, NIL
outperforms the Clifford training method in QEM performance.

We further evaluate the performance of NIL in combination with the ZNE
protocol~\cite{kim2023evidence}.
For a fair comparison, NIL chooses all the neighbor circuits as ZNE neighbors
(see \cref{sec:methods} for details), using the same 5,000 training circuits generated
previously.
The resulting linear model is then applied to the 1,000 non-Clifford test
circuits generated previously.
For the ZNE approach, we apply the same protocol as in
Ref.~\cite{kim2023evidence} (see \cref{sec:methods} for details) to the same set of test
circuits and compute the MSE between the mitigated and the ideal expectation
values.

The comparison results are summarized in \cref{tab:zne_comparison}.
It can be seen that NIL with ZNE neighbors consistently outperforms the original
ZNE protocol across all the test circuits, further demonstrating the power of
our approach.
It is worth noting that, as a flexible machine-learning-based approach, NIL
enables further improvement of the ZNE protocol by incorporating different types
of neighbor circuits.
A detailed comparison and application of this approach are provided in the
Appendix~H.

\begin{table}[h]
  \centering
  \begin{small}
  \begin{tabular}{lccc}
    \toprule
    Circuit & NIL & ZNE \\
    \midrule
    $\mathsf{vqe}$-$6$-$4$  & $1.82\times 10^{-6}$  & $1.09\times 10^{-4}$  \\
    $\mathsf{vqeRy}$-$6$-$4$ & {$3.80\times 10^{-6}$}  & $1.56\times 10^{-4}$ \\
    $\mathsf{vqe}$-$3$-$2$-$2$ & {$1.31\times 10^{-6}$}  & $1.92\times 10^{-4}$  \\
    $\mathsf{vqe}$-$8$-$4$ & {$2.63\times 10^{-6}$} & $1.58\times 10^{-3}$  \\
    $\mathsf{vqeRy}$-$8$-$4$ & {$5.78\times 10^{-6}$}  & $3.52\times 10^{-3}$ \\
    $\mathsf{vqe}$-$4$-$2$-$2$ & {$4.11\times 10^{-6}$}  & $3.63\times 10^{-3}$\\
    \bottomrule
  \end{tabular}    
  \end{small}
	\caption{Comparison of NIL and ZNE performance.
    Test MSE for NIL using ZNE neighbors versus the standard ZNE protocol on
    various VQE circuits.
    NIL consistently achieves an error reduction of two orders of magnitude.}%
  \label{tab:zne_comparison}
\end{table}

\begin{figure*}[htb]
  \captionsetup[subfigure]{justification=centering, labelfont = bf}
  \centering
  \begin{subfigure}[htb]{0.33\textwidth}
    \centering
    \includegraphics[width = \linewidth]{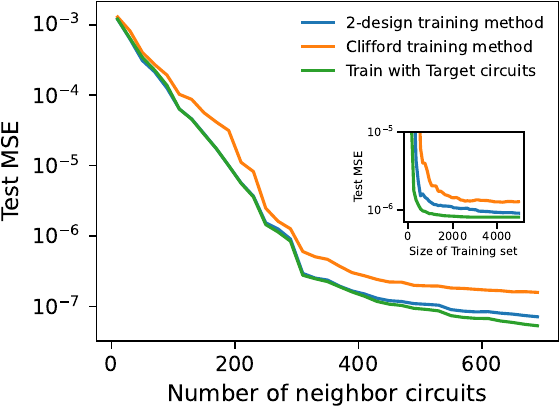}
    \subcaption{$\mathsf{vqe}$-6-4 (6 qubits, 13 layers)}
  \end{subfigure}%
  \hfill
  \begin{subfigure}[htb]{0.33\textwidth}
    \centering
    \includegraphics[width = \linewidth]{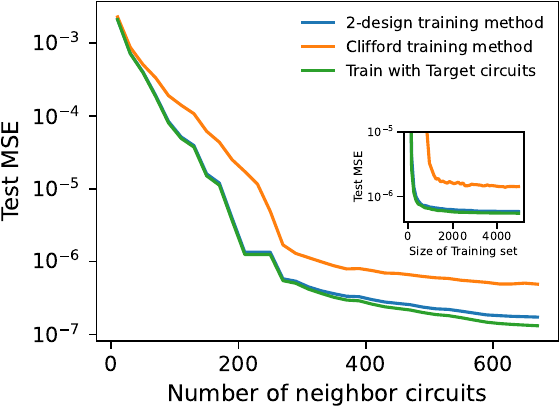}
    \subcaption{$\mathsf{vqe}$-$R_y$-$6$-$4$ (6 qubits, 17 layers)}
  \end{subfigure}%
  \hfill
  \begin{subfigure}[htb]{0.33\textwidth}
    \centering
    \includegraphics[width = \linewidth]{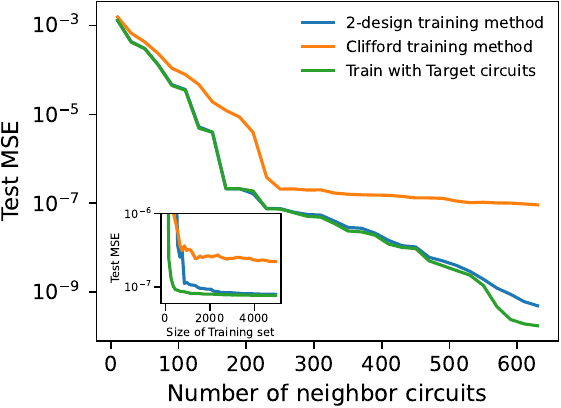}
    \subcaption{$\mathsf{vqe}$-(3, 2)-2 (6 qubits, 20 layers)}
  \end{subfigure}%

  \vspace{\baselineskip}

  \begin{subfigure}[htb]{0.33\textwidth}
    \centering
    \includegraphics[width = \linewidth]{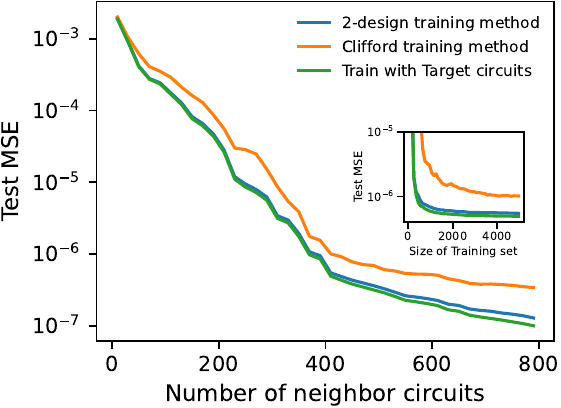}
    \subcaption{$\mathsf{vqe}$-$8$-$4$ (8 qubits, 13 layers)}
  \end{subfigure}%
  \hfill
  \begin{subfigure}[htb]{0.33\textwidth}
    \centering
    \includegraphics[width = \linewidth]{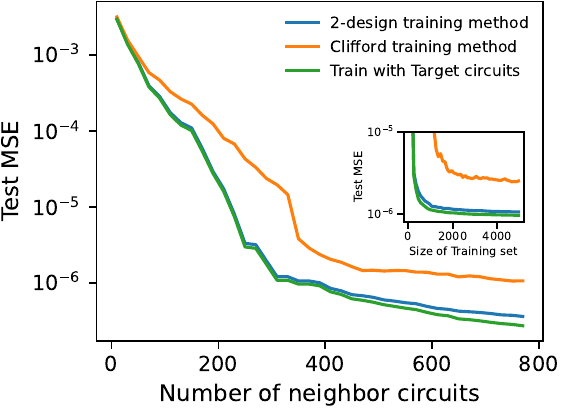}
    \subcaption{$\mathsf{vqe}$-$R_y$-$8$-$4$ (8 qubits, 17 layers)}
  \end{subfigure}%
  \hfill
  \begin{subfigure}[htb]{0.33\textwidth}
    \centering
    \includegraphics[width = \linewidth]{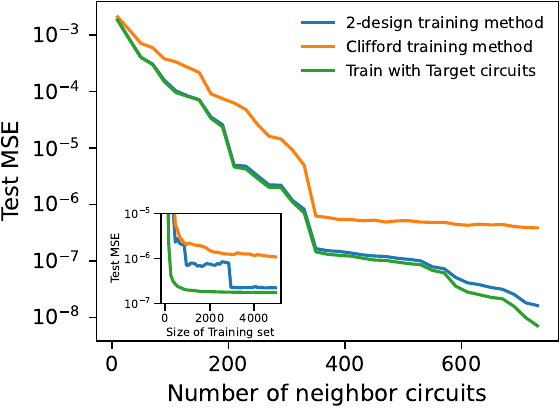}
    \subcaption{$\mathsf{vqe}$-(4, 2)-2 (8 qubits, 20 layers)}
  \end{subfigure}%

  \caption{Performance of learning-based QEM protocols using different sets of
    training circuits.
    The blue curve shows the MSE of the solution obtained via NIL, the green
    curve corresponds to the solution from direct fitting on the test dataset,
    and the orange curve represents the solution using the Clifford training
    method.
    The inset illustrates the decrease in MSE as the size of the training set
    increases, with neighbor circuits selected as all weight-$1$ neighbors.}%
  \label{fig:weight123_24_9_val}
\end{figure*}

\paragraph{UCC ansatz circuits for LiH.}

We also consider a quantum chemistry problem, i.e., we choose the UCC ansatz to
study the LiH and $\text{F}_2$ molecules~\cite{peruzzo_variational_2014} (see
Appendix~I for the circuit diagram).
As in~\cite{guo_experimental_2024}, we adopt a series of simplification
strategies for this problem, such as imposing the point group symmetry
constraint.
The description of the Hamiltonian and the ansatz circuit after simplification
can be found in Appendix~I.
Compared to the TFI Hamiltonian, molecular Hamiltonians from quantum chemistry
are generally more complex, where the number of terms scales as
$O(n^4)$~\cite{inoue_almost_2024}.

For this case, we first employ Pauli-insertion neighbors to compare the
performance of NIL with that of the Clifford training
method~\cite{strikis_learning-based_2021}, following the same procedure as
before.
Since this circuit has only four parameterized rotation gates, NIL requires a
training dataset of size $4^4 = 256$.
In sharp contrast, a complete training dataset of size $24^4$ is required for
the Clifford training method.
Similarly, we generate 1,000 non-Clifford circuits to construct a test set for
evaluating the learned linear estimators.
In addition, we apply Lasso regression directly to these 1,000 test circuits to
obtain a reference linear estimator, aiming to verify whether our method can
learn a model that converges to the actual optimal one.

As shown in \cref{fig:Performance_vqelih6}, the MSE of the linear model obtained
by our method (blue curve) on the test dataset decreases continuously to a value
close to zero as the number of neighbor circuits increases.
Furthermore, its MSE curve almost overlaps with that of the linear estimator
directly fitting the test dataset (green curve), clearly demonstrating the
strong generalization capability and the excellent QEM performance of our
approach.
Remarkably, as we can see in the inset of \cref{fig:Performance_vqelih6}, NIL
achieves a 4 order magnitude improvement over the Clifford training method when
the models employ all the weight-1 neighbor circuits.

\begin{figure}[htb]
  \centering
  \includegraphics[width = 0.4\columnwidth]{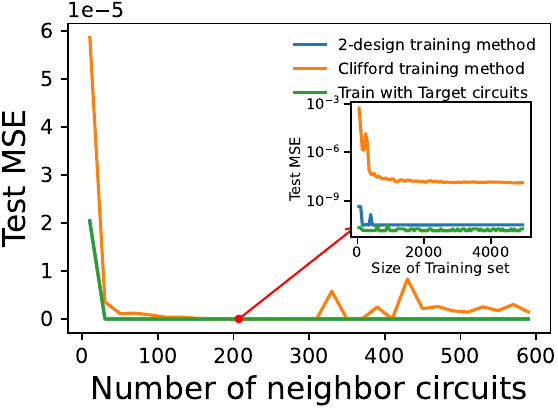}
  \caption{Performance of learning-based QEM protocols using different methods
    for generating training sets.
    The underlying circuit is the 6-qubit UCC ansatz circuit for LiH.}%
  \label{fig:Performance_vqelih6}
\end{figure}

Another crucial observation here is that as the number of neighboring circuits
increases, our approach consistently leads to better QEM performance.
This occurs because our solution converges to the optimal linear estimators, and
thus more neighbor circuits naturally reduce the obtained MSE on the target
circuit.
However, this is not necessarily true for the Clifford training method.
As shown by the orange curve in \cref{fig:Performance_vqelih6}, the performance
of the Clifford training method is not even stable.
This confirms that the results given by the Clifford training method contain
systematic errors.
In fact, we find that the linear estimator trained using the Clifford training
method actually converges to the optimal linear QEM strategy for a different
class of target circuits (see Appendix~G for details), which
explains the observed bias.

Finally, we compare the performance of the ZNE protocol and that of NIL on this
problem.
The ZNE implementation again follows Ref.~\cite{kim2023evidence}, while NIL is
trained using the 2-design method with ZNE neighbors, followed by Lasso
regression to obtain a linear estimator.
For the 1,000 test circuits generated previously, the ZNE method yields an MSE
of $1.23 \times 10^{-4}$, whereas NIL with ZNE neighbors achieves an MSE of
$2.22 \times 10^{-8}$---a four-order-of-magnitude improvement.

Based on our data in this case, achieving chemical accuracy (approximately
$6.6 \times 10^{-4}$~\cite{pople1999nobel,kirkpatrick2021pushing}, corresponding
to a squared error below $4.3 \times 10^{-7}$) requires only four ZNE neighbor
circuits to construct the neighbor map.
This example illustrates that for quantum circuits with relatively few
parameters, our method delivers excellent QEM performance with a small training
cost.
Furthermore, for the 12-qubit UCC ansatz of the $\mathrm{F}_2$ molecule, we
apply our NIL method and find that using only half of the weight-1 Pauli
neighbors suppresses about 99\% of the errors.
Further details are provided in Appendix~I.

\section{Scalability to Large Quantum Circuits}\label{subapp:perform_large_circuit}

A major challenge for learning-based QEM is benchmarking its performance on
circuits that are too large to simulate classically.
Our framework addresses this through a crucial theoretical result (Supplementary
Appendix~D), which guarantees that the average MSE on the training set is
identical to that on the test set.
This equality allows us to reliably benchmark NIL's performance on classically
intractable systems by evaluating it on the efficiently simulable Clifford
training circuits.

We first test our approach on the $\mathsf{vqe}$-(5,4)-2 circuit, which is for
the 2D transverse-field Ising model and consists of 20 qubits and 26 layers.
We employ all weight-$1$ Pauli neighbors as the set of neighbor circuits and
generate 1,000 training circuits using the 2-design training method.
Then we calculate the noisy expectation value for each neighbor circuit with
$10,000$ shots (see \cref{sec:methods} for details) and apply Lasso regression to fit the
training dataset.
The results are shown in \cref{subfig:performance_vqe20}.
It can be seen that the training MSE ($\approx$ test MSE) for this 20-qubit
circuit is very small, demonstrating that our method works well in this case.
\begin{figure}[htb!]
  \captionsetup[subfigure]{position=top, justification=raggedright, singlelinecheck=off, labelfont = bf}
  \centering
  \begin{subfigure}{0.4\textwidth}
    \centering
    \includegraphics[width = \linewidth]{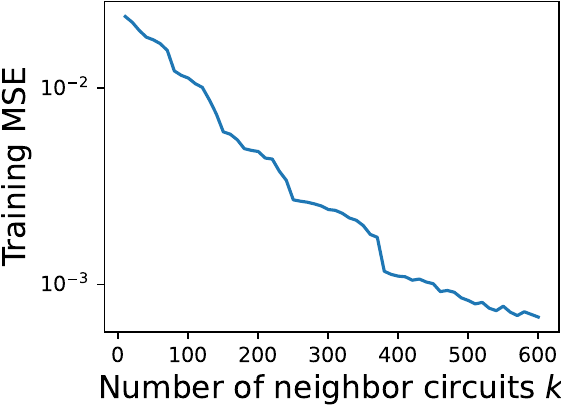}
    \subcaption{The $\mathsf{vqe}$-(5,4)-2 circuit (20 qubits, 26 layers).
      In this case, the MSE of the noisy outputs without mitigation is 0.15.}%
    \label{subfig:performance_vqe20}
  \end{subfigure}%
  \hfill
  \begin{subfigure}{0.4\textwidth}
    \centering
    \includegraphics[width = \linewidth]{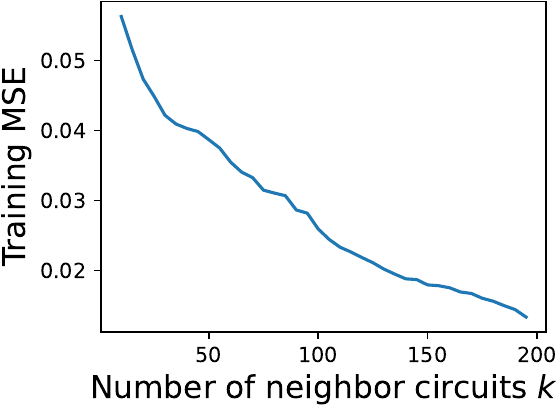}
    \subcaption{The $\mathsf{vqe}$-$R_y$-100-5 circuit (100 qubits, 21 layers).
      In this case, the MSE of the noisy outputs without mitigation is 0.41.}%
    \label{subfig:performance_vqeRy100}
  \end{subfigure}%
  \caption{{Performance of NIL on large-scale quantum
      circuits}.}\label{fig:large_circuits}
\end{figure}

Lastly, we evaluate the performance of our approach on a 100-qubit circuit
$\mathsf{vqe}$-$R_y$-100-5, which consists of 21 layers.
Here we use at most 200 random weight-1 Pauli neighbors and generate 1,000
training circuits via the 2-design training method.
The training procedure remains the same as in previous experiments.
As shown in \cref{subfig:performance_vqeRy100}, our method achieves an excellent
performance, further validating the effectiveness of our approach on large-scale
quantum circuits.

Notably, the overall sampling overhead for these three cases amounts to at most
$6\times 10^9$ shots, which represents a reasonable cost for state-of-the-art
superconducting quantum platforms, as estimated in~\cite{aharonov2025on}.
With this overhead, we achieve at least one order of magnitude of error
mitigation improvement.
In particular, for the 12-qubit UCC ansatz circuit of $\text{F}_2$, the
mitigation gain approaches nearly two orders of magnitude.
These results demonstrate that NIL can deliver excellent error mitigation
performance on large-scale circuits at experimentally feasible costs.

\section{Implementation Details}\label{sec:methods}
\paragraph{Pauli insertion neighbor circuits.}
For a given target quantum circuit, inspired by PEC, a natural strategy to
construct neighbor circuits is to perturb the original circuit by inserting a
set of gates from a certain set $\GG$ into $C$.
For example, when inserting a gate $P\in\GG$ after the single gate $U_j$, it
will be replaced by $P U_j$, while the rest of the circuit remains unchanged.
To recover the ideal expectation value, a learning-based QEM protocol needs to
collect the noisy expectation values from all circuits ${\{\noise{\C}_{j}\}}_j$,
which are obtained by inserting gates from $\GG$ after specific gates in the
original circuit $\C$.

For Pauli insertion neighbor circuits, $\GG := \{ \X, \Y, \Z \}$, the same as
the learning-based PEC approach~\cite{strikis_learning-based_2021}.
In the language of NIL, the original PEC approach uses a number of neighbor
circuits that is exponential in the number of noise positions.
To make the error mitigation scheme resource-efficient,
Ref.~\cite{strikis_learning-based_2021} proposed a truncation strategy that only
keeps the low-weight neighbor circuits, where a weight-$\ell$ neighbor circuit
is constructed by inserting quantum gates after at most $\ell$ positions of the
original circuit.
In what follows, we also generate neighbor circuits by inserting gates after
noisy operations, and we similarly restrict our attention to the low-weight
case.
Specifically, weight-$m$ Pauli neighbors refer to circuits where a Pauli gate is
inserted after at most $m$ noisy gates.
For a quantum circuit $C$ with $m$ possible noise positions, the set of weight-1
neighbor circuits is of size $O(m)$, and the set of weight-2 neighbor circuits
is of size $O(m^2)$.
In practice, the number of weight-2 neighbor circuits can already be
prohibitive, even for modest-sized quantum circuits with $m \sim 10^3$ gates.
Therefore, we need a more flexible strategy to choose neighbor circuits than the
simple truncation strategy.

For this purpose, we introduce a
randomized approach for selecting neighbor circuits.
For instance, to construct a scheme with $s$ neighbor circuits, where $s$ is
smaller than the total number of all weight-1 neighbor circuits, we randomly
select $s$ circuits from the weight-1 neighbor circuits as the chosen set.
This approach enables the error mitigation scheme to adapt to the tolerable
experimental cost.

\paragraph{Lasso regression.}
Lasso regression~\cite{tibshirani1996regression} is a special form of linear
regression that adds an $\ell_1$-norm constraint on the coefficients that need
to be optimized.
Suppose the training set is
\begin{equation*}
  \mathbb{S} = {\left\{ \left( \left( x_1^{(i)}, \ldots, x_N^{(i)} \right),
      y^{(i)} \right) \right\}}_{i=1}^T,
\end{equation*}
then Lasso regression solves the following optimization problem:
\begin{align}\label{eq:lasso_regression}
  \begin{split}
    \min_{c_1, \ldots, c_N}
    & \ \frac{1}{T} \sum_{i=1}^T {\left( \sum_{j=1}^N c_j x_j^{(i)} - y^{(i)} \right)}^2 \\
    \text{subject to} &\ \sum_{j=1}^N |c_j| \leqslant \gamma .
  \end{split}
\end{align}
where $\gamma > 0$ is the regularization.
In our numerical calculations, we solve the above convex optimization problem
using \texttt{CVXPY}~\cite{tibshirani1996regression}, with
\texttt{MOSEK}~\cite{mosek} as the backend solver due to its high precision and
robustness for large-scale convex programs.

To control the sparsity of the learned coefficients, we use Lasso regression
with different regularization strengths depending on the type of neighbor
circuits.
Specifically, we set $\gamma = 5$ when using ZNE-based neighbor circuits, and
$\gamma = 2$ for all other neighbor circuit constructions.
This choice helps to reduce the $\ell_1$ norm of the solution while maintaining
strong QEM performance.

\paragraph{ZNE implementation.}
The ZNE protocols considered include exponential extrapolation and linear
extrapolation, both implemented via \texttt{mitiq}~\cite{mitiq}.
As noted in~\cite{kim2023evidence}, exponential extrapolation suffers from
numerical instability in some cases.
When this is the case, we fall back to using linear extrapolation.

\paragraph{ZNE neighbor circuits.}
We also propose a neighbor map construction method inspired by ZNE\@.
Recall that in ZNE protocols, neighbor circuits are obtained by amplifying all
noise channels, i.e., effectively replacing each noise channel $\Lambda$ with
$\Lambda^{\alpha}$.
In this work, following the noise extrapolation approach
of~\cite{kim2023evidence}, we choose $\alpha \in \{1, 1.1, 1.34, 1.58\}$. 
We refer to the neighbor circuits constructed using this method as \emph{ZNE neighbor circuits}. In subsequent work~\cite{zhang2025quantum}, we systematically compare the performance under different noise amplification rates and demonstrate that our approach substantially improves the accuracy of variational quantum algorithms.

\paragraph{Expectation value estimation.}
For large-scale quantum circuits, we estimate the expectation values of
observables by executing the circuit multiple times and averaging the
measurement outcomes.
To reduce experimental overhead and facilitate implementation, we employ Pauli
basis measurements.
Specifically, the target observable is first decomposed into a linear
combination of Pauli operators.
For example, the TFI Hamiltonian can be written as
\begin{equation}
H_{\rm TFI} = -J \sum_{(i, j) \in E} Z_i Z_j - h \sum_{j \in V} X_j.
\end{equation}
We then group mutually commuting Pauli terms so that each group can be measured
simultaneously using a single measurement basis.
In the case of the TFI Hamiltonian, the decomposition yields two commuting
groups: ${\{Z_i Z_j\}}_{(i,j)\in E}$ and ${\{X_j\}}_{j\in V}$.
To estimate the overall expectation value, we allocate an equal number of
measurement shots to each group.
For instance, using 10,000 total shots, we assign 5,000 shots to each group.
Since the training circuits considered in our study are Clifford circuits
subject to Pauli noise, all measurement data can be efficiently generated using
Clifford simulators \texttt{stim}~\cite{gidney2021stim}.

\section{Discussions}
We have introduced NIL, a general and scalable framework for learning-based QEM\@.
At its core is a principled strategy for constructing sets of training circuits
based on quantum 2-designs.
We have also proven that this strategy yields the optimal linear estimator in an
average sense.
By employing Lasso regression, we are able to reduce the quantum sampling cost
and show a theoretical understanding for scalability of NIL\@: to achieve an
estimator with MSE within $\varepsilon$ of the optimum with confidence
$1-\delta$, it suffices to use $\mathcal{O}(\ln(N/\delta)/\varepsilon^2)$
training circuits, where $N$ is the number of features (neighbor circuits).
This scaling makes the method efficient and realistic even for large-scale
quantum circuits.
Numerical simulations on variational quantum circuits---including 1D and 2D
transverse-field Ising models and molecular Hamiltonians like LiH---demonstrate
the strong empirical performance of NIL\@.
We have also compared different neighbor circuit generation strategies, and find
that Pauli insertion neighbors consistently outperform CPTP-basis constructions.

Interestingly, an important and valuable feature of our framework lies in its
applicability to circuits that are beyond classical simulations.
We have shown that NIL can be reliably evaluated even on quantum circuits with
20 to 100 qubits or more, thanks to a theoretical guarantee equating the average
MSE on the set of Clifford training circuits to the MSE on target circuits,
under mild assumptions.
This enables benchmarking of learning-based QEM schemes on classically
intractable systems---an important step toward practical deployment.

While our analysis has focused on linear models and the MSE loss function,
several questions remain open.
Although we provide theoretical and numerical evidence that a low MSE implies
reliable mitigation with high probability for any specific circuit instance
(Appendix E), the relationship between average-case and worst-case
performance warrants further investigation.
Additionally, while our tests with neural networks did not show an advantage
over linear models (Appendix H), the potential benefits of more
sophisticated non-linear models in different noise regimes or for different
tasks remain an interesting direction for future research.

In conclusion, NIL provides a theoretically grounded, practically efficient, and
scalable framework for quantum error mitigation.
By unifying concepts from existing methods and leveraging a provably optimal
learning strategy, it offers a viable path toward extracting reliable results
from today's noisy quantum hardware and advancing the pursuit of quantum
advantage.


\appendix
\onecolumn
\input{sup.tex}

\bibliographystyle{unsrt}
\bibliography{main}


\end{document}

%% file: sup.tex
\section*{\centering Supplemental Material}
\appendix

\renewcommand{\contentsname}{\centering Contents}
\tableofcontents

\clearpage

\section{Implementation Details of the 2-Design Training Method}\label{app:2-design_method}

In this section, we present the detailed procedure of the 2-design training method.
Suppose that the target quantum circuit is $C$, and a set of training quantum
circuits $S = \{ C^{(1)}, C^{(2)}, \ldots, C^{(T)} \}$ is constructed, referred
to as ``training circuits''.
For each training circuit $C^{(i)}$, the chosen neighbor circuits is
\begin{equation*}
 \neighbor \bigl(C^{(i)} \bigr) = \bigl(C^{(i)}_{1}, \ldots, C^{(i)}_{N} \bigr).
\end{equation*}
For each $i=1, 2, \ldots, T$ and $j = 1, 2, \ldots, N$, we run the circuit
$\C^{(i)}_{j}$ and obtain the noisy target expectation value
$\bignavg{\C^{(i)}_{j}}$.
Then, the whole training data for $\mathbb{S}$ can be written as
\begin{equation*}
  \mathbb{S} := {\left\{\left(\left(x_1^{(i)}, \ldots, x_N^{(i)}\right)^{\top},
        y^{(i)}\right)\right\}}_{i=1,\dots,T},
\end{equation*}
where
\begin{equation}\label{eq:features}
  x^{(i)}_j = \Bignavg{\C^{(i)}_{j}} ,
\end{equation}
and
\begin{equation}\label{eq:labels}
  y^{(i)} = \avg{\C^{(i)}}.
\end{equation}
That is, for each training quantum circuit $C^{(i)}$ in $S$, the final
training data includes the noisy expectation values of all its neighbors, and
the noiseless expectation value $\avg{\C^{(i)}}$ itself.
Then, one needs to find a proper function by which the
value of $\avg{\C^{(i)}}$ can be predicted well based on those of
$\bignavg{\C^{(i)}_{j}}$ for all the training data.
The hope is that machine learning can capture the pattern of quantum noise,
enabling the learned function to predict the ideal expectation value as well on the
target quantum circuit $C$ by analyzing the noisy expectation values of all its
neighboring circuits.


A key requirement in selecting the training circuits $C^{(i)}$ is ensuring that
the noiseless expectation value $\langle C^{(i)} \rangle$ can be computed both
precisely and efficiently.
A common strategy for this purpose is to choose Clifford training circuits, as
they have the desirable property that the exact expectation values
$\langle C^{(i)} \rangle$ can be computed
efficiently~\cite{gottesman_heisenberg_1999, aaronson_improved_2004}.

In learning-based QEM, proper construction of training data is crucial for the
success of machine learning models.
One of our main contributions in this work is the discovery of an effective
strategy, called \emph{$2$-design training}, that generates training quantum
circuits based on the mathematical tool of $t$-design~\cite{gross_evenly_2007}.
In particular, we demonstrate that this strategy enables us to efficiently
identify the optimal $\combine$ function among all linear options, outperforming
even neural networks in terms of performance.

\subsection{Target Quantum Circuits}\label{subapp:test_set}

Before introducing our new approach, we would like to clarify the target quantum
circuits (sometimes referred to as ``test circuits'') for our QEM task, which
are a family of quantum circuits sharing the same structure.
These circuits can be written as
$C(\vb*{\theta}) = U_m(\theta_m) \cdots U_1(\theta_1)$, where
$\vb*{\theta} = (\theta_1,\ldots,\theta_m)$, and each $\theta_i$ lies in the
interval $[0, 2\pi]$.
Here, we assume that each unitary operation $U_i(\theta_i)$ is either a
single-qubit parameterized rotation gate $R_{\sigma}(\theta_i)$ with the
rotation angle $\theta_i$ or a fixed one- or two-qubit Clifford gate $F_i$
without parameters, where $\sigma = x, y, z$ corresponds to the rotation axis.
See \cref{fig:diagram_target_circuit} for an example.
Additionally, after fixing the circuit structure, we assume that the noise
models of all quantum gates remain unchanged when $\theta_i$'s are varied.

\begin{figure}[htbp!]
  \centering
  \includegraphics[width = 0.4\columnwidth]{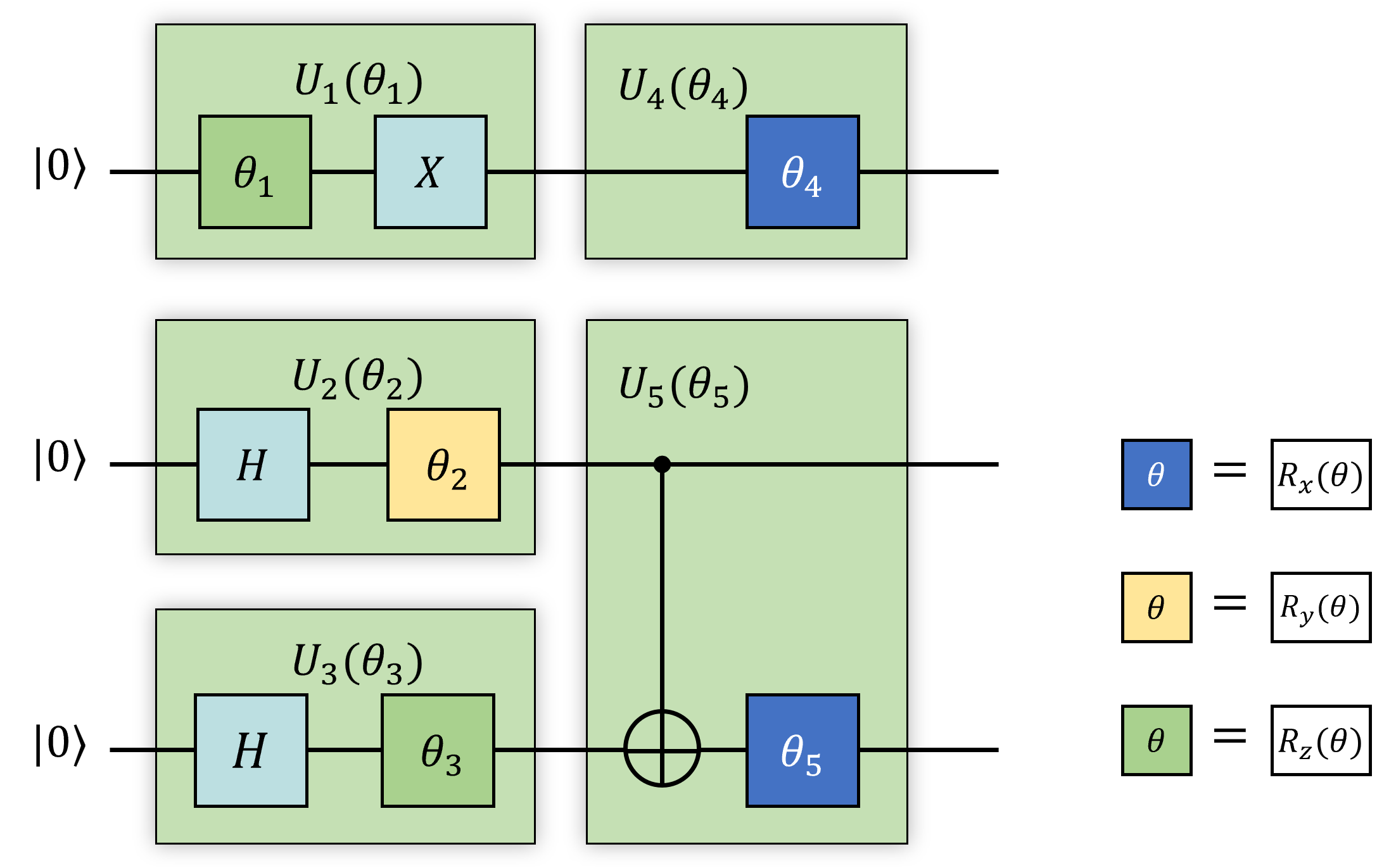}
  \caption{A typical parameterized quantum circuit $C(\vb*{\theta})$.
    This circuit contains six parameters and can be represented in the form of
    $C(\vb*{\theta}) = U_5(\theta_5) \cdots U_1(\theta_1)$, where each
    $U_i(\theta_i)$ is composed of several gates enclosed within a light green
    frame.
    For instance,
    $U_1(\theta_1)= XR_{z}(\theta_1)$.}\label{fig:diagram_target_circuit}
\end{figure}

This type of quantum circuit is very common for near-term quantum algorithms
such as variational quantum algorithms~\cite{peruzzo_variational_2014,
  farhi_quantum_2014}, and covers nearly all hardware-efficient variational
quantum circuits.
In a typical variational quantum algorithm, one often constructs an ansatz
circuit $C(\vb*{\theta})$ and adjusts the parameters $\vb*{\theta}$ such that
the output of the circuit can approximate the ground state of a Hamiltonian $H$.
Since optimal solutions may appear at any value of $\vb*{\theta}$, our objective
is to achieve effective error mitigation for all potential values of
$\vb*{\theta}$ at a tolerable cost, rather than to focus on a specific
individual circuit or a specific value of $\vb*{\theta}$.

\subsection{Data Generation Method}\label{subapp:training_set}

Previous works typically construct the set of training quantum circuits by
maintaining the structure of the target quantum circuit while replacing each
(single-qubit) non-Clifford gate with a Clifford gate uniformly sampled from all
24 single-qubit Clifford gates~\cite{strikis_learning-based_2021}
(Clifford training method), or from a subset of these Clifford
gates~\cite{czarnik2021error}.
We refer to this strategy for constructing training quantum circuits as the
\emph{Clifford training method}, which is a heuristic approach without rigorous
performance analysis.

In contrast, the \emph{$2$-design training} method (2-training method) we
propose replaces each non-Clifford gate with one of the \emph{four} gates
forming a rotational $2$-design, as will be elaborated shortly.
Recall that the single-qubit parameterized gates in the target circuits are
Pauli rotation gates $R_{\sigma}(\theta)$ with $\sigma = x, y, z$.
The four Clifford gates for each gate $R_{\sigma}(\theta)$ are
$R_{\sigma}(\theta_{0})$, where $\theta_{0}$ is chosen from
$\{0, \pi/2, \pi, 3\pi/2\}$.
Besides significantly reducing the computational cost of the training stage, we
will demonstrate that the effectiveness of the $2$-design training method has a
solid theoretical foundation in \cref{app:optimality}, and its performance is numerically verified to
significantly surpass that of the Clifford training method.

\begin{algorithm}[htb!]
  \textbf{Input:} Target parameterized circuit $C(\vb*{\theta})$\\
  \textbf{Output:} A set $S$ of training circuits
  \begin{algorithmic}[1]
    \Repeat%
      \For{each gate in the circuit $C(\vb*{\theta})$}
        \If{the gate is of the form $R_{\sigma}(\theta)$ ($\sigma = x,y,z$)}
          \State{Replace $R_{\sigma}(\theta)$ with $g\in \mathbb{C}_1$ randomly,
            where $\mathbb{C}_1$ is the single-qubit Clifford group.}
        \EndIf{}
      \EndFor{}
      \State{Add the new circuit to $S$}
    \Until{Collect sufficient circuits for training}
  \end{algorithmic}
  \caption{Strategy for constructing training circuits with all single-qubit
    Clifford gates}\label{alg:gen_training_all_clifford}
\end{algorithm}

\begin{algorithm}[htb!]
  \textbf{Input:} Target parameterized circuit $C(\vb*{\theta})$\\
  \textbf{Output:} A set $S$ of training circuits
  \begin{algorithmic}[1]
    \Repeat%
      \For{each gate in the circuit $C(\vb*{\theta})$}
        \If{the gate is of the form $R_{\sigma}(\theta)$ ($\sigma = x,y,z$)}
          \State{$\theta \gets_{\$} \{0, \pi/2, \pi, 3\pi/2\}$}
        \EndIf{}
      \EndFor{}
      \State{Add the new circuit to $S$}
    \Until{Collect sufficient circuits for training}
  \end{algorithmic}
  \caption{Strategy for constructing training circuits with quantum rotation
    $2$-design}\label{alg:gen_training}
\end{algorithm}

In 2-training method, ``$\gets_{\$}$'' means ``sampled uniformly at random
from''.
We set the size of the training set $T$ to be sufficiently large so that the
loss function (defined later) remains relatively stable during training and
converges eventually.
In practice, we can gradually increase the size of $T$, monitor the loss
function fluctuations, and add more training circuits if needed.

Additionally, we consider a training set generation method that produces
training circuits with a shallow non-Clifford layer, bringing the training
circuits closer to the test circuit.
See \cref{app:mixed_non_clf} for more details.


\subsection{Loss Function and Learning Models}\label{subapp:learning_algorithms}

After preparing the training set, we need to train an appropriate function
$\combine$ that fits the training set.
The loss function is chosen to be the mean squared error (MSE), which is a
widely used metric to quantify the performance of machine learning models and is
defined as
\begin{equation}\label{MSE}
  {\rm MSE} := \frac{1}{T}\sum_{i=1}^{T}{\left(y^{(i)}-\hat{y}^{(i)}\right)}^2.
\end{equation}
Here, $y^{(i)}$ is the noiseless expectation value of the $i$-th training
circuit, $\hat{y}^{(i)}$ is the corresponding value predicted by the model, and
$T$ is the size of the training set.


In this work, we analyze and compare three learning models designed to minimize
the MSE on the training set.
We first consider two linear models for $\combine$: linear
regression and Lasso regression.
Specifically, given the training dataset $\mathbb{S}$, a linear regression
solves the following optimization problem:
\begin{equation}\label{eq:linear_regression}
  \min _{c_1, \ldots, c_N} \frac{1}{T} \sum_{i=1}^{T}
  {\left(\sum_{j=1}^N c_j x_j^{(i)}-y^{(i)}\right)}^2,
\end{equation}
where $c_j$'s are the parameters to learn.
We solve this problem using the \texttt{LinearRegression} function from
\texttt{sklearn}~\cite{scikit-learn}.

It is straightforward to see that \cref{eq:linear_regression} is a Least Squares
(LS) problem~\cite{bjorck1996numerical}.
However, the $\ell_1$ norm of the solution obtained from
\cref{eq:linear_regression} can be too large, making the obtained $\combine$
impractical.
To overcome this problem, we use Lasso
regression~\cite{tibshirani1996regression}, which imposes an $\ell_1$ norm
constraint on the parameters:
\begin{align}
\label{eqapp:lasso_regression}
\begin{split}
  \min _{c_1, \ldots, c_N} &\ \frac{1}{T} \sum_{i=1}^{T}
  {\left(\sum_{j=1}^N c_j x_j^{(i)}-y^{(i)}\right)}^2 \\
   \text{subject to } &\  \sum_i |c_i| \leqslant \gamma,
\end{split}
\end{align}
where $\gamma > 0$.
In our numerical calculations, we solve the convex optimization problem in
\cref{eqapp:lasso_regression} using \texttt{CVXPY}~\cite{tibshirani1996regression},
with \texttt{MOSEK}~\cite{mosek} as the backend solver due to its high precision
and robustness for large-scale convex programs.

Using Lasso regression, we can reduce the $\ell_1$ norm of the solution by
appropriately choosing $\gamma$.
For instance, in linear regression, the value of $\sum_i |c_i|$ can reach as
high as $10^4$ in some cases we studied, whereas with Lasso regression, it can
be reduced to less than 5 while the MSE only increases slightly.
As previously discussed, this reduction in the $\ell_1$ norm implies that the
computational cost of the QEM protocol will decrease by orders of magnitude.
Another advantage of Lasso regression is its tendency to set many coefficients
to zero, meaning that in practical applications, we can omit constructing
neighboring circuits of the target circuit for those coefficients that vanish.
\section{Optimality of the $2$-Design Training Method}\label{app:optimality}

The mathematical properties of the $2$-design training method are analyzed
below.
Specifically, we will prove that the function $\combine$ obtained by the linear
regression or the Lasso regression based on the training set generated by
2-training method is, in some sense, optimal among all linear functions as
$T \to \infty$.
To achieve this, we first define a concept called ``\emph{quantum rotation
  $t$-design}'' and show that the gate set produced in 2-training method forms a
quantum rotation $2$-design.
Then, we will prove that only $\mathcal{O}(\ln{(N/\delta)}/\varepsilon^2)$
training circuits are needed to obtain a solution within $\varepsilon$ distance
from the optimal linear solution {with probability at least $1-\delta$}, where
$N$ is the number of neighbor circuits.
Lastly, we will argue that achieving a small MSE on the training set guarantees
that the worst-case error on the test set remains small with high probability.

\subsection{Quantum Rotation $t$-design}\label{subapp:rotation_design}

Let $R_{\hat{n}}(\theta) := \exp (-i \theta \hat{n} \cdot \va{\sigma} / 2)$,
where $\hat{n} = \left(n_x, n_y, n_z\right)$ is a unit vector and
$\va{\sigma} = (X, Y, Z)$.
It is not hard to see that the set
$\mathcal{G}_{\hat{n}} := {\{ R_{\hat{n}}(\theta) \}}_{\theta \in [0, 2\pi]}$
forms a group, which is a subgroup of the single-qubit unitary group
$\mathbb{U}(2)$.
Similar to unitary $t$-design~\cite{gross_evenly_2007}, here we consider the
$t$-design over the group $\mathcal{G}_{\hat{n}}$, which we call the
\emph{quantum rotation $t$-design}.

\begin{definition}
  A set of  unitary matrices ${\left\{A_i\right\}}_{i=1}^K$ on $\CC^{2\times 2}$
  is called a \emph{quantum rotation $t$-design} with respect to
  $R_{\hat{n}}(\theta)$, if
  \begin{align*}
    \frac{1}{K} \sum_{i=1}^K {\left(A_i \otimes A_i^{\dagger} \right)}^{\otimes t} =
    \frac{1}{2\pi} \int_{0}^{2\pi} {\left(R_{\hat{n}}(\theta)
    \otimes R_{\hat{n}}(-\theta) \right)}^{\otimes t} d\theta.
  \end{align*}
\end{definition}

We now prove that the gate set utilized in 2-training method actually
forms a quantum rotation $2$-design.
\begin{theorem}\label{thm:two_design}
  ${\left\{R_{\hat{n}}(\theta)\right\}}_{\theta=0, \pi/2, \pi, 3\pi/2}$ is a
  quantum rotation $2$-design with respect to
  ${\{ R_{\hat{n}}(\theta) \}}_{\theta \in [0, 2\pi]}$.
\end{theorem}

\begin{proof}
  Define the matrix $M$ as $M = \hat{n} \cdot \va{\sigma}$, and then
  $R_{\hat{n}}(\theta) = \exp (-i \theta M / 2)$.
  Using the relations
  \begin{align*}
    &\frac{1}{2 \pi} \int_0^{2 \pi} \cos ^4 \frac{\theta}{2} {\rm d}\theta
      = \frac{1}{2 \pi} \int_0^{2 \pi} \sin ^4 \frac{\theta}{2} {\rm d}\theta
      = \frac{3}{8},\\
    & \frac{1}{2 \pi} \int_0^{2 \pi} \cos \frac{\theta}{2}
      \sin ^3 \frac{\theta}{2} {\rm d}\theta
      = \frac{1}{2 \pi} \int_0^{2 \pi} \cos ^3 \frac{\theta}{2}
      \sin \frac{\theta}{2} {\rm d}\theta=0, \\
    & \frac{1}{2 \pi} \int_0^{2 \pi} \cos ^2 \frac{\theta}{2}
      \sin ^2 \frac{\theta}{2} {\rm d}\theta=\frac{1}{8},
  \end{align*}
  we have
  \begin{align*}
    & \frac{1}{2\pi}\int_{0}^{2\pi} {R_{\hat{n}}(\theta)}^{\otimes 2} \otimes
      {R_{\hat{n}}(-\theta)}^{\otimes 2} {\rm d}\theta \\
    = & \sum_{i_1,\ldots, i_4=0}^1 \frac{1}{2 \pi} \int_0^{2 \pi} i^{-i_1-i_2+i_3+i_4}
        {\Bigl(\cos \frac{\theta}{2}\Bigr)}^{4-\sum_j i_j}\times\\
    & \qquad \qquad {\Bigl(\sin \frac{\theta}{2}\Bigr)}^{\sum_j i_j}
      \times \biggl(\bigotimes_{j=1}^{4} M^{i_{j}} \biggr)\,
      {\rm d}\theta\\
    = & \frac{3}{8} I^{\otimes 4} + \frac{3}{8} M^{\otimes 4} +
        \frac{1}{8}  \sum_{\substack{i_1, \ldots, i_4=0 \\
      i_1+\cdots+i_4=2}}^1 i^{-i_1-i_2+i_3+i_4}
      \bigotimes_{j=1}^{4} M^{i_{j}}.
  \end{align*}

  Meanwhile, it can be verified that
  \begin{align*}
    & \frac{1}{4} \sum_{k=0}^3 {R_{\hat{n}} \Bigl( \frac{k\pi}{2}
      \Bigr)}^{\otimes 2} \otimes {R_{\hat{n}} \Bigl(\frac{-k\pi}{2}
      \Bigr)}^{\otimes 2} \\
    = & \frac{1}{4} \left(1+\frac{1}{4}+\frac{1}{4}\right)I^{\otimes 4}+
        \frac{1}{4} \left(1+\frac{1}{4}+\frac{1}{4}\right)M^{\otimes 4}\\
    & + \frac{1}{4} \left(\frac{1}{4}+\frac{1}{4}\right)
       \sum_{\substack{i_1, \ldots, i_4=0 \\
      i_1+\cdots+i_4=2}}^1 i^{-i_1-i_2+i_3+i_4}
      \bigotimes_{j=1}^{4} M^{i_{j}}\\
    = & \frac{3}{8} I^{\otimes 4} + \frac{3}{8} M^{\otimes 4} +
        \frac{1}{8}  \sum_{\substack{i_1, \ldots, i_4=0 \\
      i_1+\cdots+i_4=2}}^1 i^{-i_1-i_2+i_3+i_4}
      \bigotimes_{j=1}^{4} M^{i_{j}},
  \end{align*}
  which concludes the proof.
\end{proof}

Now we are ready to prove the optimality of the solutions obtained by linear
regression and Lasso regression on the training circuits generated by
2-training method.

\subsection{{Asymptotic Optimality}}\label{app:prove_optimal}

We prove that under the 2-design training method, both linear regression and
Lasso regression achieve optimal solutions on the training set as its size $T$
approaches infinity.
Interestingly, these solutions coincide with the optimal solutions on the test
set on average.
This demonstrates the strong generalization capability of the 2-design training
method.
We first analyze the idealized setting without quantum measurement shot noise,
then extend our proof to account for shot noise effects.

\subsubsection{Optimality without Shot Noise}

Firstly, recall that the training set is
\begin{equation*}
  \mathbb{S} = {\left\{\left(\left(x_1^{(i)}, \ldots, x_N^{(i)}\right)^{\top},
        y^{(i)}\right)\right\}}_{i=1,\ldots,T},
\end{equation*}
where the feature vector
$\vb{x}^{(i)} = \bigl(x_1^{(i)}, \ldots, x_N^{(i)}\bigr)^{\top}$ records the
noisy expectation values for the $N$ neighbor circuits of the $i$-th training
circuit {$C^{(i)}$}, $y^{(i)}$ represents the ideal expectation value for
this training circuit, and $T$ is the number of training circuits.
We denote by $P_{\mathbb{S}}$ the underlying probability distribution of the
training circuit $C^{(i)}$, which is determined by how we generate the training
set $\mathbb{S}$, i.e., how the single-qubit Pauli rotation gates are replaced
by Clifford gates.
In the 2-design training method, the rotation angles of the Pauli gates are randomly reset to one of four distinct values.

With a slight abuse of notation, we also denote $P_{\mathbb{S}}$ as the
probability distribution of all types of data generated from the training
circuit $C^{(i)}$, for example, the feature vector $\vb{x}^{(i)}$ and the label
$y^{(i)}$.
Similarly, we define $P_{\rm test}$ as the underlying distribution of the target
circuit, where all parameters in the parameterized circuits are sampled
uniformly in $[0, 2\pi)$.

\paragraph{Linear regression.}
The goal of a linear regression model is to learn a linear function such that, for each training
circuit, it can simultaneously map the $N$ noisy expectation values from
neighbor circuits to the ideal expectation value, i.e., the mapping has the form
  $\hat{y}=\sum_{j=1}^N c_j x_j$,
where $c_j$'s are the parameters we need to learn.
For this, we aim to minimize the average MSE on the training set
\begin{equation*}
  \min _{c_1, \ldots, c_N} \frac{1}{T} \sum_{i=1}^{T}
  {\left(\sum_{j=1}^N c_j x_{j}^{(i)} - y^{(i)}\right)}^2 .
\end{equation*}
When the size of the training set tends to infinity, i.e., $T \to \infty$, the
solution we obtained is actually the solution to the following optimization
problem:
\begin{equation}\label{eq:training_problem}
  \min _{c_1, \ldots, c_N} \underset{\left(\left(x_1, \ldots, x_N\right), y\right)
    \sim P_{\mathbb{S}}}{\Exp}{\left(\sum_{j=1}^N c_j x_j-y\right)}^2.
\end{equation}
We now formulate \cref{eq:training_problem} as a quadratic programming problem.
Define
$\vb{c}:={(c_1, c_2, \ldots, c_N)}^{\top}, \vb{x}:={(x_1, x_2, \ldots, x_N)}^{\top}$,
then \cref{eq:training_problem} can be rewritten as
\begin{align*}
  & \min_{\vb{c}} \ExyS {\left(\vb{c}^{\top}  \vb{x} - y \right)}^2 \\
  = & \min_{\vb{c}} \ExyS {\left(\vb{c}^{\top}  \vb{x} - y\right)
      \left(\vb{c}^{\top} \vb{x} - y\right)}^{\top}\\
  = & \min_{\vb{c}} \vb{c}^{\top} \ExyS
      \left[\vb{x} \vb{x}^{\top} \right] \vb{c} - 2\vb{c}^{\top}
      \ExyS \left[y\vb{x}\right] + \ExyS y^2.
\end{align*}
Since the coefficient matrix $\ExyS \left[\vb{x} \vb{x}^{\top}\right]$ is
positive semidefinite, \cref{eq:training_problem} is a quadratic programming
problem with a semidefinite coefficient matrix~\cite{boyd2004convex}, whose
solution is given by
\begin{align}\label{eq:solution_train}
  \vb{c}^* = {\begin{pmatrix}
    \ES x_1 x_1
    & \cdots
    & \ES x_1 x_N \\
    \vdots
    & \ddots
    & \vdots \\
    \ES x_N x_1
    & \cdots
    & \ES x_N x_N
  \end{pmatrix}}^{-1}
  \begin{pmatrix}
    \ES x_1 y \\
    \vdots \\
    \ES x_N y
  \end{pmatrix},
\end{align}
where we write $\ExyS$ as $\ES$ for simplicity.
If $\ExyS \left[\vb{x} \vb{x}^{\top} \right]$ is not invertible, then the
pseudoinverse should be used instead~\cite{penrose1955generalized}.
For the convenience of later discussion, we introduce the following notations:
\begin{align}\label{eq:coefficient_matrix}
  \vb{A} =
  \begin{pmatrix}
     x_1 x_1 & \cdots &  x_1 x_N \\
    \vdots & \ddots & \vdots\\
     x_N x_1 & \cdots &x_N x_N
  \end{pmatrix},\;
  \vb{a} =
  \begin{pmatrix}
    x_1 y \\
    \vdots \\
    x_N y
  \end{pmatrix}.
\end{align}
Then \cref{eq:solution_train} becomes
$\vb{c}^* = {\left(\ES [\vb{A}] \right)}^{-1} \ES [\vb{a}]$.

After training, we hope that the obtained model can effectively mitigate errors
in the target parameterized quantum circuits.
Next, we show that using 2-training method to construct the training set,
the optimal solution $\vb{c}^*$ obtained on the training set is also an optimal
solution on the test set.

To see this, we first estimate the smallest average MSE that a linear model can
possibly achieve on the target parameterized quantum circuits, which is the
solution to the optimization problem
\begin{equation}\label{eq:test_problem}
  \min_{\vb{c}} \Exyt {\left(\vb{c}^{\top} \vb{x} - y\right)}^2,
\end{equation}
where $P_{\rm test}$ is the underlying distribution of the test dataset.
Here, each data point also includes the noisy expectation values of all the
neighbor circuits and the noiseless expectation value of the target circuit.
Then the solution to \cref{eq:test_problem} is actually
\begin{align}\label{eq:solution_test}
  {\begin{pmatrix}
    \Et x_1 x_1
    & \cdots & \Et x_1 x_N \\
    \vdots & \ddots & \vdots \\
    \Et x_N x_1
    & \cdots & \Et x_N x_N
  \end{pmatrix}}^{-1}
  \begin{pmatrix}
    \Et x_1 y \\
    \vdots \\
    \Et x_N y
  \end{pmatrix},
\end{align}
which can also be written as ${\left(\Et [\vb{A}] \right)}^{-1} \Et [\vb{a}]$.
Here, we have used the shorthand notation $\Et$ for $\Exyt$.

The following lemma builds a key connection between the data of computing
produced by the training circuits and that by the target circuits.
\begin{lemma}\label{lemma:two_design}
  Suppose the training circuits are generated by
  2-training method.
  If we apply the same neighbor map to the training circuits and the target
  quantum circuit, we will have that
  $\ES\left[\vb{A}\right] = \Et\left[\vb{A}\right]$,
  $\ES\left[\vb{a}\right]=\Et\left[\vb{a}\right]$, {and $\ES\left[y^2\right] = \Et\left[y^2\right]$}.
\end{lemma}

\begin{proof}
  Assume that $U_1$ is a single-qubit gate acting on the first qubit {of
    neighbor circuits, and the corresponding rotation angle is $\theta_1$}.
  Let $\Phi_1$ ($\Phi_1^{\prime}$) denote the noisy channel corresponding to the
  operations of all the gates before $U_1$ in the $i$-th ($j$-th) neighbor
  circuit.
  Define $\rho=\Phi_1\left(\rho_{\text {in }}\right)$ and
  $\rho' = \Phi'_1(\rho_{\rm in})$, where $\rho_{\text {in }}$ is the same input
  state of all the circuits.
  Similarly, let $\Phi_2$ ($\Phi_2^{\prime}$) be the channel corresponding to
  the operations of all the gates after $U_1$ in the $i$-th ($j$-th) neighbor
  circuit.
  See \cref{fig:i_j_neighbor} for an illustration.

  Suppose
  $U_1 = R_{\hat{n}}(\theta_1) = \exp (-i \theta_1 M / 2)$, where $M = \hat{n} \cdot \va{\sigma}$. Then
  \begin{align*}
    & \underset{U_1}{\Exp} \tr \left(\left( U_1 \rho U_1^{\dagger} \right)
      \Phi_2^{\dagger}(O)\right)
      \tr\left(\left(U_1 \rho^{\prime} U_1^{\dagger} \right)
      \Phi_2^{\prime\dagger}(O)\right) \\
    = & \underset{\theta_1 \in[0,2\pi]}{\Exp}
        \tr \Bigl( \bigl(R_{\hat{n}}(\theta_1)
        \rho R_{\hat{n}}(-\theta_1) \bigr) \Phi_2^{\dagger}(O) \Bigr)
        \tr\Bigl(\bigl( R_{\hat{n}}(\theta_1) \rho^{\prime} R_{\hat{n}}(-\theta_1) \bigr) \Phi_2^{'\dagger}(O)\Bigr) \\
    = &\underset{\theta_1 \in\left\{0, \frac{\pi}{2}, \pi, \frac{3}{2} \pi\right\}}{\Exp}
        \tr \Bigl( \bigl(R_{\hat{n}}(\theta_1)
        \rho R_{\hat{n}}(-\theta_1) \bigr) \Phi_2^{\dagger}(O) \Bigr)
        \tr\Bigl(\bigl( R_{\hat{n}}(\theta_1) \rho^{\prime} R_{\hat{n}}(-\theta_1) \bigr) \Phi_2^{'\dagger}(O)\Bigr),
  \end{align*}
  where the last equality holds because
  ${\left\{R_{\hat{n}}(\theta_1)\right\}}_{\theta_1=0, \pi/2, \pi, 3\pi/2}$
  forms a quantum rotation $2$-design.
  Suppose all the other parameterized gates in the circuits are
  $U_2, U_3, \ldots, U_M$ with the corresponding rotation angles
  $\theta_2, \theta_3, \ldots, \theta_M$, which are contained in $\Phi_2$ and
  $\Phi_2'$.
  It follows that
  \begin{align*}
    &\underset{P_{\text{test}}}{\Exp}\, x_i x_j
     = \underset{U_M}{\Exp} \cdots \underset{U_1}{\Exp}\, x_i x_j\\
     =& \underset{\theta_{1}, \ldots, \theta_{m} \in [0,2\pi]}{\Exp}
      \tr\Bigl( \bigl( R_{\hat{n}}(\theta_1) \rho R_{\hat{n}}(-\theta_1) \bigr)
      \Phi_2^{\dagger}(O)\Bigr) \tr\Bigl(\bigl( R_{\hat{n}}(\theta_1) \rho^{\prime}
      R_{\hat{n}}(-\theta_1) \bigr) \Phi_2^{'\dagger}(O)\Bigr)\\
     =& \underset{\substack{\theta_2, \ldots, \theta_{M} \in[0,2\pi]\\
    \theta_1 \in\left\{0, \frac{\pi}{2}, \pi, \frac{3}{2} \pi\right\}}}{\Exp}
    \tr\Bigl(\bigl( R_{\hat{n}}(\theta_1) \rho R_{\hat{n}}(-\theta_1) \bigr)
    \Phi_2^{\dagger}(O) \Bigr) \tr\Bigl(\bigl( R_{\hat{n}}(\theta_1) \rho^{\prime}
      R_{\hat{n}}(-\theta_1) \bigr) \Phi_2^{'\dagger}(O)\Bigr)\\
     =& \underset{\substack{\theta_2, \ldots, \theta_{M} \in[0,2\pi]\\
    \theta_1 \in\left\{0, \frac{\pi}{2}, \pi, \frac{3}{2} \pi\right\}}}{\Exp}
    x_i x_j.
  \end{align*}
  We conclude the proof by repeating the above procedure for
  $\theta_2, \theta_3, \ldots, \theta_M$ sequentially.

  \begin{figure}[tb!]
    \centering
    \includegraphics[width = 0.4\columnwidth]{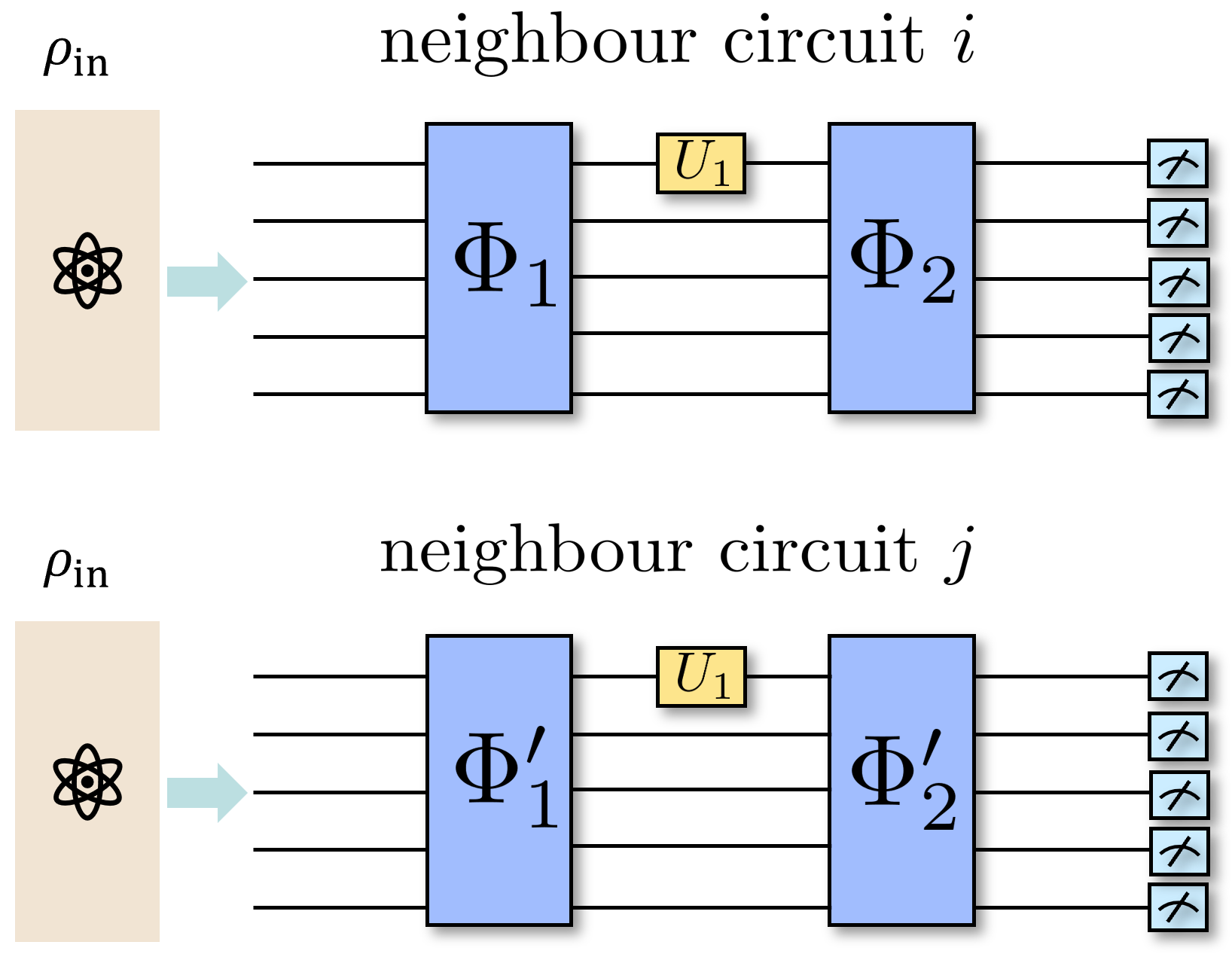}
    \caption{The $i$-th and the $j$-th neighbor circuits}\label{fig:i_j_neighbor}
  \end{figure}

\end{proof}

Together with \cref{eq:solution_train} and \cref{eq:solution_test},
\cref{lemma:two_design} leads to the following conclusion.

\begin{theorem}[Optimality, linear regression]
\label{thm:optimal}
  Suppose the training set
  \begin{equation*}
    \mathbb{S} = {\left\{ \left( \left(x_1^{(i)}, \ldots,
            x_N^{(i)} \right)^{\top}, y^{(i)} \right) \right\}}_{i=1,\ldots,T}
  \end{equation*}
  is produced by 2-training method.
  Then the solution obtained from the linear regression in
  \cref{eq:linear_regression} {is also a} solution to the linear regression
  \cref{eq:test_problem} on the target circuit when $T \to \infty$.
\end{theorem}


\paragraph{Lasso regression.}
Recall that to reduce the computational cost of learning-based QEM, one can
replace linear regression with Lasso regression to fit the training data.
Below, as a counterpart to \cref{thm:optimal}, we prove that
2-training method also enables Lasso regression to achieve optimal
solutions on both the training set and test set.

\begin{theorem}[Optimality, Lasso regression]\label{thm:convergence_lasso}
  The solution to the Lasso regression on the training circuit set generated by
  2-training method is also a solution for the target parameterized
  circuits as the size of the training set $T \to \infty$.
\end{theorem}

\begin{proof}
  Recall that the optimization problem that the Lasso regression tries
  to solve on the training set is
  \begin{align}\label{eq:lasso_appendix}
    \min _{c_1, \ldots, c_N} &\frac{1}{T} \sum_{i=1}^{T}
    {\left(\sum_{j=1}^N c_j x_j^{(i)}-y^{(i)}\right)}^2 \\
     \text{subject to } & \sum_i |c_i| \leqslant \gamma.\nonumber
  \end{align}
  When $T \to \infty$, the above problem can also be written as the following
  optimization problem:
  \begin{align*}
    \min_{\vb{c}}\ & \vb{c}^{\top} \ES\left[\vb{A}\right] \vb{c} - 2\vb{c}^{\top}
    \ES\left[\vb{a}\right] + \ES y^2\\
     \text {subject to } & \|\vb{c}\|_1 \leq \gamma.
  \end{align*}
  According to \cref{lemma:two_design}, its solution is also the solution to the
  following optimization problem:
  \begin{align*}
    \min_{\vb{c}}\ & \vb{c}^{\top} \Et\left[\vb{A}\right] \vb{c} - 2\vb{c}^{\top}
    \Et\left[\vb{a}\right] + \Et y^2\\
     \text {subject to } & \|\vb{c}\|_1 \leq \gamma,
  \end{align*}
  which is exactly the Lasso regression on the target parameterized circuit.
\end{proof}

\subsubsection{{Optimality with Shot Noise}}\label{subapp:optimal_shot_noise}

In this subsection, we discuss the impact of shot noise on our protocol.
Previously, we have proven that our method yields the optimal linear solutions
if we apply either linear regression or Lasso regression on a training set
generated by 2-training method, where the noisy expectation values
obtained are exact.
However, in practice, all the noisy expectation values in the feature vectors
are obtained by quantum measurements, which introduce statistical errors for
each expectation value $x_j^{(i)}$ defined in \cref{eq:features}.
Specifically, if we measure each neighbor circuit $N_s$ times and average the
outcomes to get the expectation value, the actual training set takes the
following form:
\begin{align}\label{eq:shots_training_set}
  \mathbb{S}^{N_s} = {\left\{\left(\vb{x}^{(i)}_{N_s},y^{(i)}\right)\right\}}_{i\in [T]},
\end{align}
where $\vb{x}^{(i)}_{N_s}$ denotes the feature vector of the $i$-th training circuit.
The $j$-th element of $\vb{x}_{N_s}^{(i)}$ is
\begin{align}
  \left(\vb{x}^{(i)}_{N_s} \right)_j = \frac{1}{N_s}\sum_{k=1}^{N_s} x_{j,k}^{(i)} \ ,
\end{align}
where $x_{j,k}^{(i)}$ denotes the $k$-th measurement outcome on the
$j$-th neighbor circuit of the $i$-th training circuit.
As $N_s$ increases, we have $\vb{x}^{(i)} = \lim_{N_s \to \infty} \vb{x}^{(i)}_{N_s}$.

We now prove that the linear (Lasso) regression problem defined in
\cref{eq:linear_regression} (\cref{eqapp:lasso_regression}) on the training set
$\mathbb{S}^{N_s}$ is equivalent to the linear (Lasso) regression problem with
an additional $\ell_2$ regularization term defined on the test
dataset~\cite{ng2004feature}, when $T\to \infty$.
In particular, the strength of this regularization term is inversely
proportional to $N_s$.
In many linear regression protocols, an $\ell_2$ regularization term is added to
avoid overfitting~\cite{hoerl1970ridge} and improve numerical
stability~\cite{ng2004feature,golub1999tikhonov}.
Interestingly, in our task, shot noise naturally introduces an $\ell_2$ penalty
term, making the optimization problem more stable and easier to solve.

To prove this statement, we require the following lemma, which establishes the
relationship between the average variance of the measurement outcomes on the
training set and that on the test set.


\begin{lemma}\label{lemma:same_variance}
  Let $C$ be randomly sampled from the test circuits and let $C^{(i)}$ be a
  training circuit generated by the 2-design method (2-training method).
  Apply the same neighbor map to $C$ and $C^{(i)}$ to obtain the corresponding
  neighbor circuits $(C_1, \ldots, C_N)$ and $(C_1^{(i)}, \ldots, C_N^{(i)})$,
  respectively.
  Denote by $x_{j, 1}^{(i)}$ and $o_{j, 1}$ the measurement outcomes of
  measuring the observable $O$ on the $j$-th neighbor circuit of $C^{(i)}$ and
  $C$, respectively.
  Then, we have
  \[
  \ES\left[\Var[x_{j,1}^{(i)}]\right] = \Et\left[\Var[o_{j,1}]\right]
  \]
  for any $j\in [N]$, where the randomness of $\,\Var$ is from quantum
  measurements.
\end{lemma}

\begin{proof}
  The proof technique is similar to that in the proof for
  \cref{lemma:two_design}.
  Suppose that the noisy output state of the $j$-th neighbor circuit of
  $C^{(i)}$ is $\rho$.
  Then we have
  \[
  \Var[x_{j,1}^{(i)}] = \tr(\rho O^2) - \tr(\rho O)^2.
  \]
  As in the proof for \cref{lemma:two_design}, we examine
  $\ES\left[\Var[x_{j,1}^{(i)}]\right]$ gate by gate.
  Suppose that
  $\rho=\Phi_2(R_{\hat{n}}(\theta_1)\\\Phi_1\left(\rho_{\text {in }}\right)R_{\hat{n}}(-\theta_1))$,
  where $\rho_{\text {in }}$, $R_{\hat{n}}(\theta_1)$, and $\Phi_i$ ($i=1,2$)
  are defined similarly as in \cref{lemma:two_design}.
  Then it holds that
  \begin{align*}
    &\underset{\theta_1 \in\left\{0, \frac{\pi}{2}, \pi, \frac{3}{2} \pi\right\}}{\Exp} \tr(\rho O^2)\\
    = &\underset{\theta_1 \in\left\{0, \frac{\pi}{2}, \pi, \frac{3}{2} \pi\right\}}{\Exp}
        \tr \left( \bigl(R_{\hat{n}}(\theta_1)
        \Phi_1\left(\rho_{\text {in }}\right) R_{\hat{n}}(-\theta_1) \bigr) \Phi_2^{\dagger}(O^2) \right)\\
    = &\underset{\theta_1 \in[0,2\pi]}{\Exp}
        \tr \left( \bigl(R_{\hat{n}}(\theta_1)
        \Phi_1\left(\rho_{\text {in }}\right) R_{\hat{n}}(-\theta_1) \bigr) \Phi_2^{\dagger}(O^2) \right),
  \end{align*}
  where the last equality holds because
  ${\left\{R_{\hat{n}}(\theta)\right\}}_{\theta=0, \pi/2, \pi, 3\pi/2}$ is a
  quantum rotation $2$-design, implying that it is also a quantum rotation
  1-design.
  Taking the expectation over all random rotation gates, denoted as
  $U_1, \ldots, U_M$, we obtain that
  \begin{align*}
    \ES\tr(\rho O^2)&= \underset{U_M}{\Exp} \cdots \underset{U_1}{\Exp}\tr(\rho O^2)\\
    &= \underset{
    \theta_1\ldots \theta_M \in\left\{0, \frac{\pi}{2}, \pi, \frac{3}{2} \pi\right\}}{\Exp}\tr \left(\bigl(R_{\hat{n}}(\theta_1)
    \Phi_1\left(\rho_{\text {in }}\right) R_{\hat{n}}(-\theta_1) \bigr) \Phi_2^{\dagger}(O^2) \right)\\
    &= \underset{\substack{ \theta_1\in[0,2\pi]\\
     \theta_2, \ldots, \theta_{M}\in\left\{0, \frac{\pi}{2}, \pi, \frac{3}{2} \pi\right\}}}{\Exp}\tr \left(\bigl(R_{\hat{n}}(\theta_1)
     \Phi_1\left(\rho_{\text {in }}\right) R_{\hat{n}}(-\theta_1) \bigr) \Phi_2^{\dagger}(O^2) \right),
  \end{align*}    %
  where the parameters $\theta_2, \ldots, \theta_M$ are contained in $\Phi_2$.
  Repeating the above procedure for $U_2, \ldots, U_M$ sequentially yields that
  \begin{align*}
    \ES\tr(\rho O^2)&= \underset{U_M}{\Exp} \cdots \underset{U_1}{\Exp}\tr(\rho O^2)\\
    &= \underset{\theta_1\ldots, \theta_{M}\in[0,2\pi]}{\Exp}\tr \left(\bigl(R_{\hat{n}}(\theta_1)
     \Phi_1\left(\rho_{\text {in }}\right) R_{\hat{n}}(-\theta_1) \bigr) \Phi_2^{\dagger}(O^2) \right)\\
     &= \Et \tr \left(\bigl(R_{\hat{n}}(\theta_1)
     \Phi_1\left(\rho_{\text {in }}\right) R_{\hat{n}}(-\theta_1) \bigr) \Phi_2^{\dagger}(O^2) \right) \\
     &= \Et \tr(\rho' O^2),
  \end{align*}%
  where we denoted the output state of the $j$-th neighbor of the target circuit as $\rho'$.
  According to \cref{lemma:two_design}, we know that $\ES \tr(\rho O)^2 = \Et \tr(\rho' O)^2$.
  Putting them together, we have
  \begin{align*}
    \ES\left[\Var[x_{j,1}^{(i)}]\right]&= \ES \left[\tr(\rho O^2) - \tr(\rho O)^2\right]\\
    &= \Et \left[\tr(\rho' O^2) - \tr(\rho' O)^2\right]\\
     &= \Et\left[\Var[o_{j,1}]\right].
  \end{align*}
\end{proof}

Now we are ready to prove the asymptotic optimality of the 2-design training method in the presence of
shot noise.
Specifically, we will show that the shot noise introduces an additional $\ell_2$
regularization term to the optimization problem we need to address.
Recall that
\begin{align*}
  \vb{A} = \vb{x}\vb{x}^\top,\;
  \vb{a} =
  \begin{pmatrix}
    x_1 y \\
    \vdots \\
    x_N y
  \end{pmatrix},
\end{align*}
where $\vb{x}=\left(x_1, \ldots, x_N\right)^{\top}$.
For linear regression, we have the following theorem.

\begin{theorem}\label{thm:optimal_shot_noise}
  Consider the training dataset
  \begin{equation*}
    \mathbb{S}^{N_s}  = {\left\{\left(\left(\frac{1}{N_s}\sum_{j=1}^{N_s} x_{1,j}^{(i)},
            \ldots, \frac{1}{N_s}\sum_{j=1}^{N_s} x_{N,j}^{(i)}\right)^{\top},
          y^{(i)}\right)\right\}}_{i\in [T]} \ ,
  \end{equation*}
  where the training circuits are generated with 2-training method.
  When $T \to \infty$, the solution obtained from the linear regression problem
  in \cref{eq:linear_regression} is also a solution to the following regression
  problem on the target circuit:
  \begin{equation*}
    \min_{\vb{c}} \vb{c}^{\top} \Et\left[\vb{A} + \ReNs\right] \vb{c} - 2 \Et\left[\vb{a}\right]^{\top} \vb{c} + \Et y^2,
  \end{equation*}
  The matrix \( \ReNs \) is defined as
\begin{equation*}
  \ReNs \coloneqq \frac{\mathrm{diag}(X_1, \ldots, X_N)}{N_s},
\end{equation*}
where $X_i$ is the variance of observable $O$ associated with the $i$-th
neighbor circuit of the target circuit.
\end{theorem}
\begin{proof}
  We first introduce some random variables $\varepsilon_{i,j}$'s to represent
  the shot noise caused by quantum measurements.
  Specifically, we define
  \[
    \varepsilon_{i,j} := \frac{1}{N_s}\sum_{k=1}^{N_s} x_{j,k}^{(i)} - x_{j}^{(i)},
  \]
  for $i \in [T]$ and $j \in [N]$.

  Next, we investigate the statistical properties of $\varepsilon_{i,j}$.
  Denote $\Em$ as the expectation with respect to the measurement outcomes.
  Then, conditioned on the training circuit $C^{(i)}$ being sampled, we have
  \begin{align*}
  \Em[ \varepsilon_{i,j} | C^{(i)}] = \frac{1}{N_s}\sum_{k=1}^{N_s} x_{j}^{(i)} - x_{j}^{(i)} =  0.
  \end{align*}
  Similarly, we compute its second moment:
  \begin{align*}
    &\Em[{\varepsilon_{i,j}^2} | C^{(i)}]\\
    =& \Em\left[\frac{1}{N_s^2}\left(\sum_{k=1}^{N_s} \left(x_{j,k}^{(i)}\right)^2 + \sum_{k_1\neq k_2}x_{j,k_1}^{(i)}x_{j,k_2}^{(i)}\right)\right.\left. - \frac{2x_{j}^{(i)}}{N_s}\sum_k x_{j,k}^{(i)} + \left(x_{j}^{(i)}\right)^2 \Big| C^{(i)} \right]\\
    =&\Em\left[\frac{1}{N_s^2}\sum_{k=1}^{N_s} \left(x_{j,k}^{(i)}\right)^2 \Big| C^{(i)} \right] + \frac{(N_s-1)}{N_s}\left(x_{j}^{(i)}\right)^2 - \left(x_{j}^{(i)}\right)^2\\
    =&\Em \left[\frac{1}{N_s}\left(x_{j,1}^{(i)}\right)^2 \Big| C^{(i)} \right]-\frac{1}{N_s}\left(x_{j}^{(i)}\right)^2 \\
    =&\frac{1}{N_s}\Var\left[x_{j,1}^{(i)}\right] .
  \end{align*}%
  Here, the penultimate equality holds because $x_{j,k}^{(i)}$ is sampled from
  the same probability distribution independently for any $k\in[N_s]$.
  We now rewrite the feature vector of the $i$-th training circuit as
  \begin{align*}
    \vb{x}^{(i)}_{N_s}= &\left(\frac{1}{N_s}\sum_{j=1}^{N_s} x_{1,j}^{(i)}, \ldots, \frac{1}{N_s}\sum_{j=1}^{N_s} x_{N,j}^{(i)}\right)^{\top}\\
    = &\left(x_{1}^{(i)} + \varepsilon_{i,1}, \ldots, x_{N}^{(i)} + \varepsilon_{i,N}\right)^{\top}.
  \end{align*}
  Then, the linear regression defined on the training set $\mathbb{S}^{N_s}$ can be expressed as
  \begin{align*}
    & \min _{\vb{c}} \frac{1}{T} \sum_{i=1}^{T}
      {\left( \vb{c}^{\top} \vb{x}^{(i)}_{N_s}-y^{(i)}\right)}^2\\
      = & \min_{\vb{c}} \vb{c}^{\top} \Bigl(\frac{1}{T} \sum_{i=1}^{T} \vb{x}^{(i)}_{N_s} \vb{x}_{N_s}^{(i) \top}\Bigr) \vb{c} -
      2 \frac{1}{T} \sum_{i=1}^{T} y^{(i)} \vb{x}_{N_s}^{(i) \top} \vb{c} + \frac{1}{T} \sum_{i=1}^{T} {\bigl(y^{(i)}\bigr)}^2.
  \end{align*}
  Treat $\vb{x}^{(i)}_{N_s} \vb{x}_{N_s}^{(i) \top}$, $y^{(i)} \vb{x}_{N_s}^{(i) \top}$ and $\left(y^{(i)}\right)^2$ as random variables. Then, as $T\to \infty$, the linear regression defined on training set $\mathbb{S}^{N_s}$ reads
  \begin{equation*}
   \min_{\vb{c}} \vb{c}^{\top} \Etrain\left[\vb{x}^{(i)}_{N_s}\vb{x}_{N_s}^{(i) \top}\right] \vb{c} -2 \Etrain[y^{(i)}\vb{x}^{(i)}_{N_s}]^{\top}\vb{c}+\Etrain \left[\left(y^{(i)}\right)^2\right] ,
  \end{equation*}
  where the expectation $\Etrain := \ES \Em$ is taken over two sources of randomness, namely, the probability distribution of the training data point $P_{\mathbb{S}}$ and quantum measurements.

  Next, we compute $\Etrain\left[\vb{x}^{(i)}_{N_s}\vb{x}_{N_s}^{(i) \top}\right]$, $\Etrain[y^{(i)}\vb{x}^{(i)}_{N_s}]^{\top}$ and $\Etrain \left[\left(y^{(i)}\right)^2\right]$.
  Based on this observation and the law of total expectation, we have that when $j\neq k$,
  \begin{align*}
    &\left(\Etrain\left[\vb{x}^{(i)}_{N_s}\vb{x}_{N_s}^{(i) \top}\right]\right)_{j,k} \\
    = &\ES\left[\Em\left[\left(x_{j}^{(i)} + \varepsilon_{i,j}\right)\left(x_{k}^{(i)} + \varepsilon_{i,k}\right)\Big| C^{(i)} \right]\right]\\
    = & \ES[x_{j}^{(i)} x_{k}^{(i)}] = \Et[x_{j} x_{k}] \ ,
\end{align*}    %
  The digaonal elements of this matrix is given by
  \begin{align*}
    &\left( \Etrain\left[\vb{x}^{(i)}_{N_s}\vb{x}_{N_s}^{(i) \top}\right]\right)_{j,j} \\
    = &\ES\left[\Em\left[\left(x_{j}^{(i)} + \varepsilon_{i,j}\right)^2\Big| C^{(i)} \right]\right]\\
    = & \ES\left[\left(x_{j}^{(i)}\right)^2\right] + \ES\left[\Var\left[x_{j,1}^{(i)}\right]\right]/N_s,
  \end{align*}
  where $j\in[N]$ and the last equality holds because $\Em[{\varepsilon_{i,j}^2}] = \frac{1}{N_s}\Var\left[x_{j,1}^{(i)}\right]$. According to \cref{lemma:two_design} and \cref{lemma:same_variance}, we have that $\ES\left[\left(x_{j}^{(i)}\right)^2\right] = \Et\left[x_j^2\right]$ and $\ES\left[\Var\left[x_{j,1}^{(i)}\right]\right]= \Et [X_j]$, where $X_j \coloneqq \Var [o_{j,1}]$ is the variance of the measurement outcomes for the $j$-th neighbor circuit of the target circuit.
  Then, we have
  \[
    \left(\Etrain\left[\vb{x}^{(i)}_{N_s}\vb{x}_{N_s}^{(i) \top}\right]\right)_{j,j} =  \Et\left[x_{j}^2\right] + \Et [X_j]/N_s.
  \]
  Therefore, it holds that
  \[
    \Etrain\left[\vb{x}^{(i)}_{N_s}\vb{x}_{N_s}^{(i) \top}\right] = \Et\left[\vb{A} + \ReNs \right],
  \]
  where $\ReNs \coloneqq \frac{\mathrm{diag}(X_1, \ldots, X_N)}{N_s}$. Next we compute $\Etrain[y^{(i)}\vb{x}^{(i)}_{N_s}]$:
\begin{align*}
  &\left(\Etrain[y^{(i)}\vb{x}^{(i)}_{N_s}]\right)_{j} \\
  = &\ES\left[\Em\left[y^{(i)}\left(x_{j}^{(i)} + \varepsilon_{i,j}\right)\Big| C^{(i)} \right]\right]\\
  = & \ES[y^{(i)} x_{j}^{(i)}] = \Et[y x_{j}].
\end{align*}
 Since the labels of the training set $\mathbb{S}^{N_s}$ are computed classically, the randomness of $\left(y^{(i)}\right)^2$ only comes from the probability distribution of the training dataset $P_{\mathbb{S}}$. Then we have
 \[
  \Etrain \left[\left(y^{(i)}\right)^2\right] = \ES\left[\left(y^{(i)}\right)^2\right] = \Et\left[\left(y^{(i)}\right)^2\right],
 \]
 where the last equality can be obtained directly from \cref{lemma:two_design}.
 Therefore, when $T$ is sufficiently large, we have
\begin{align*}
  \frac{1}{T} \sum_{i=1}^{T} \vb{x}^{(i)}_{N_s} \vb{x}_{N_s}^{(i) \top} &\rightarrow \Et\left[ \vb{A} + \ReNs \right], \\
  \frac{1}{T} \sum_{i=1}^{T} y^{(i)} \vb{x}_{N_s}^{(i) \top} &\rightarrow \Et\left[\vb{a}\right]^{\top}, \\
  \frac{1}{T} \sum_{i=1}^{T} \left(y^{(i)}\right)^2 &\rightarrow \Et\left[\left(y^{(i)}\right)^2\right].
\end{align*}%
To conclude, when $T \to \infty$, the solution to the linear regression defined on the training set $\mathbb{S}^{N_s}$ is also a solution to the following regression problem on the target circuit:
  \begin{equation}\label{eq:solution_shot}
    \min_{\vb{c}} \vb{c}^{\top} \Et\left[\vb{A} + \ReNs \right] \vb{c} - 2 \Et\left[\vb{a}\right]^{\top} \vb{c} + \Et y^2.
  \end{equation}
\end{proof}

As a direct corollary of \cref{thm:optimal_shot_noise}, the solution to \cref{eq:solution_shot} is also a solution to the linear regression on the target circuit in the presence of shot noise, if we measure all the involved quantum circuits with the same number of shots $N_s$.

\begin{Corollary}\label{cor:same_test}
  As the size of the
  training set $T \to \infty$, the solution to the linear (Lasso) regression on the training
  set $\mathbb{S}^{N_s}$ generated by 2-training method is also a solution to the linear (Lasso) regression on the target circuit, where all the expectation values of the neighbor circuits are obtained through $N_s$ measurement shots.
\end{Corollary}

\begin{proof}
  We only give the proof for the case of the linear regression, and the case of the Lasso regression can be proved similarly. 

  Given the target circuit, we apply the neighbor map to generate the corresponding neighbor circuits and collect the feature data by performing measurements on each neighbor circuit for $N_s$ times. Let
  \[
    \vb{x}_{N_s}= \left(\frac{1}{N_s}\sum_{j=1}^{N_s} o_{1,j}, \ldots, \frac{1}{N_s}\sum_{j=1}^{N_s} o_{N,j}\right)^{\top}
  \]
  denote the feature vector for the target circuit, where $o_{i,j}$ is the   outcome when measuring the observable $O$ for the $j$-th time on the output of the $i$-th neighbor circuit. The linear regression problem defined on the test dataset is
  \begin{align}
    &\min_{\vb{c}} \Etest \left(\vb{c}^{\top}\vb{x}_{N_s} - y \right)^2 \nonumber \\
    =&\min_{\vb{c}} \vb{c}^{\top} \Etest\left[\vb{x}_{N_s}\vb{x}_{N_s}^{\top}\right] \vb{c} -2 \Etest[y\vb{x}_{N_s}]^{\top}\vb{c} + \Et y^2, \label{eq:solution_shot_test}
   \end{align}
   where the expectation $\Etest := \Et \Em$ is taken over the target circuits and the randomness of measurement outcomes.
   Following a similar procedure as in the proof of \cref{thm:optimal_shot_noise}, we obtain
   \begin{align*}
    \Etest\left[\vb{x}_{N_s}\vb{x}_{N_s}^{\top}\right] &= \Et\left[\vb{A} + \ReNs\right],
    \\
    \Etest[y\vb{x}_{N_s}] &=\Et\left[\vb{a}\right]^{\top}.
  \end{align*}
  Therefore \cref{eq:solution_shot_test} can also be written as
  \begin{equation*}
    \min_{\vb{c}} \vb{c}^{\top} \Et\left[\vb{A} + \ReNs\right] \vb{c} - 2 \Et\left[\vb{a}\right]^{\top} \vb{c} + \Et y^2,
  \end{equation*}
  which is identical to the linear regression defined on the training set when its size is sufficiently large.
\end{proof}

  {From the above theorem and corollary, we observe that the $\ell_2$ regularization term introduced by shot noise is given by
  \[
  \vb{c}^{\top}\Et\left[\ReNs\right] \vb{c}.
  \]
  It turns out that this regularization term helps us to achieve a better numerical stability, as the new coefficient matrix
  \[
    \Et\left[\vb{A} + \ReNs\right]
  \]
  would be less ill-conditioned. This suggests that shot noise helps the optimization process to converge more quickly.
  On the other hand, to ensure that the $\ell_2$ regularization term does not dominate the obtained MSE, the norm of $\vb{c}$ should not be too large.
  Otherwise, the required $N_s$ will become excessively large, making the needed computational cost unacceptable.} 




\section{Estimation of the Computational Cost}\label{app:efficiency}

\subsection{{Estimation of the Size of Training Set}}\label{subapp:estimation_trainingsize}
We have argued that with a sufficiently large training set, our method is
guaranteed to obtain the optimal linear solution with respect to the MSE\@.
We now provide an estimate of how large is sufficient, and our conclusion
suggests that the number of training circuits required is tolerable.
Specifically, we will prove that only
$\mathcal{O}({\ln{(N/\delta)}}/\varepsilon^2)$ training circuits suffice to
obtain a solution whose MSE is at most $\varepsilon$ away from that of the case
with no restrictions on the number of training circuits, with probability at
least $1-\delta$, where $N$ is the number of neighbor circuits.
In the following discussion, we treat the number of shots $N_s$ used to estimate
the expectation value of the observable $O$ as a hyperparameter.
Our proof holds for arbitrary $N_s$, including the case $N_s \to \infty$, where
the shot noise is negligible.
Recall that
$\vb{x}^{(i)}_{N_s}=\left(\frac{1}{N_s}\sum_{j=1}^{N_s} x_{1,j}^{(i)}, \ldots, \frac{1}{N_s}\sum_{j=1}^{N_s} x_{N,j}^{(i)}\right)^{\top}$,
which represents the feature vector of the $i$-th training circuit, and
$y^{(i)}$ is the corresponding label.
Solving the Lasso regression in \cref{eqapp:lasso_regression} defined on the
training set of size $T$, we obtain a solution
\[
\vb{c}_T^{N_s} \in \argmin_{\|\vb{c}\|_1 \leq \gamma} \vb{c}^{\top}\vb{A}_T^{N_s}\vb{c} - 2 \left(\vb{b}^{N_s}_T\right)^\top \vb{c} + Y_T,
\]
where we define
$\vb{A}_T^{N_s} := \frac{1}{T}\sum_{i=1}^T \vb{x}^{(i)}_{N_s}\vb{x}_{N_s}^{(i) \top}$,
$\vb{b}_T^{N_s} := \frac{1}{T} \sum_{i=1}^T\left(\left(\frac{1}{N_s}\sum_j x_{1,j}^{(i)}\right) y^{(i)}, \ldots, \left(\frac{1}{N_s}\sum_j x_{N,j}^{(i)}\right) y^{(i)}\right)^{\top}$,
and $Y_T := \frac{1}{T} \sum_{i=1}^T \left(y^{(i)}\right)^2$.
Since the solution to the Lasso regression may not be unique, we assume
$\vb{c}_T^{N_s}$ is one of the solutions returned by the optimization algorithm.

According to the proof for \cref{cor:same_test}, the MSE of using
$\vb{c}_T^{N_s}$ to predict the labels on the test set is given by
\[
  L(\vb{c}_T^{N_s}) := \vb{c}_T^{N_s {\top}} \vb{A}^{N_s}  \vb{c}_T^{N_s} - 2 \vb{b}^\top  \vb{c}_T^{N_s} + Y,
\]
where we have defined $\vb{A}^{N_s} := \Et\left[\vb{A}+\ReNs\right]$,
$\vb{b} := \Et\left(x_1y, \ldots, x_N y\right)^{\top}$, and $Y := \Et[y^2]$.
Now, consider an arbitrary solution obtained from \cref{eqapp:lasso_regression}
defined on the test set, i.e.,
\[
\vb{c}_*^{N_s} \in \argmin_{\|\vb{c}\|_1 \leq \gamma}\vb{c}^{\top}\vb{A}^{N_s}\vb{c} - 2 \vb{b}^\top \vb{c} + Y.
\]
Then, the MSE of applying $\vb{c}_*^{N_s}$ to perform QEM on the test dataset reads
\[
  L(\vb{c}_*^{N_s}) := \vb{c}_*^{N_s {\top}} \vb{A}^{N_s}  \vb{c}_*^{N_s} - 2 \vb{b}^\top\vb{c}_*^{N_s} + Y.
\]
Note that $L(\vb{c}_*^{N_s})$ is the infimum of the MSE that can be achieved on
the test dataset with Lasso regression when $T \to \infty$.

Next, we estimate the size of the training set $T$ required such that
\[
  L(\vb{c}_T^{N_s}) \leq L(\vb{c}_*^{N_s}) + \varepsilon
\]
with probability at least $1-\delta$.
We first present the following lemma, which characterizes the relationship
between the MSE on the training set and that on the test set using the same
coefficients $\vb{c}$.

\begin{lemma}\label{lemma:relation_MSE}
  For arbitrary $\|\vb{c}\|_1 \leq \gamma$, a training set of size
  $T\geq \ln \left(\frac{6N^2}{\delta}\right)
   \frac{{\left(12\gamma^2\|O\|^2\right)}^2}{\varepsilon^2} \in \mathcal{O}\left(\ln \left(\frac{N}{\delta}\right)\frac{1}{\varepsilon^2}\right)$ suffices to obtain $\vb{A}_T^{N_s}, \vb{b}_T^{{N_s}}$ and  $Y_t$ such that
    \begin{align*}
    L(\vb{c}) - \varepsilon/2 \leq \vb{c}^{\top}\vb{A}_T^{N_s}\vb{c} - 2 \left(\vb{b}^{N_s}_T\right)^\top \vb{c} + Y_T \leq L(\vb{c}) +\varepsilon/2
  \end{align*}
  with probability at least $1-\delta$.
 \end{lemma}

A similar result to the above lemma can be found in~\cite[Chapter 11]{mohri2018foundations}.
To make this paper self-contained, we provide a proof here.

\begin{proof}
  We begin by deriving an upper bound for the difference in MSE between the
  training set of size $T$ and the test set, using the same coefficients
  $\vb{c}$.
  \begin{align*}
    & \left|\vb{c}^{\top}\vb{A}_T^{N_s}\vb{c} - 2 \left(\vb{b}^{N_s}_T\right)^\top \vb{c} + Y_T - \left(\vb{c}^{\top}\vb{A}^{N_s}\vb{c} - 2 \vb{b}^\top \vb{c} + Y\right)\right|\\
    \leq & \left|\vb{c}^{\top}\vb{A}_T^{N_s}\vb{c} - \vb{c}^{\top}\vb{A}^{N_s}\vb{c}\right| + 2\left|\left(\vb{b}^{N_s}_T\right)^\top\vb{c} - \vb{b}^\top \vb{c}\right| + \left|Y_T - Y\right|\\
    \leq & \|\vb{c}\|_1 \left\|\left(\vb{A}_T^{N_s}-\vb{A}^{N_s}\right)\vb{c}\right\|_\infty + 2 \|\vb{c}\|_1\left\|\vb{b}^{N_s}_T - \vb{b}\right\|_\infty + \left|Y_T - Y\right|\\
    \leq & \gamma^2 \left\|\vb{A}_T^{N_s}-\vb{A}^{N_s}\right\|_{\text{max}} + 2\gamma \left\|\vb{b}^{N_s}_T - \vb{b}\right\|_\infty + \left|Y_T - Y\right|,
  \end{align*}    %
  where the max (infinity) norm of a matrix (vector) is the maximum absolute
  value of its entries.
  As in \cref{lemma:two_design} and \cref{thm:optimal_shot_noise}, each entry of
  $\vb{A}^{N_s}, \vb{b}$ and $Y$ is the expectation value of a certain random
  variable, and each corresponding entry of $\vb{A}_T^{N_s}, \vb{b}_T^{{N_s}}$
  and $Y_T$ is an approximation of the expectation value using $T$ samples.
  Meanwhile, the absolute values of these random variables are upper bounded:
  \begin{align*}
  &\left|\frac{1}{N_s}\sum_{l=1}^{N_s} x_{j,l}^{(i)}\frac{1}{N_s}\sum_{l=1}^{N_s} x_{k,l}^{(i)}\right|\leq \|O\|^2,\\
  &  \left|\frac{1}{N_s}\sum_{l=1}^{N_s} x_{j,l}^{(i)} y^{(i)}\right|\leq \|O\|^2, \left(y^{(i)}\right)^2 \leq \|O\|^2,
  \end{align*}
  for arbitrary $j,k\in[N]$ and $i\in[T]$, where $\|O\|$ denotes the spectral norm of the observable $O$. As a result, by Hoeffding's inequality we have that
    \begin{align*}
    &\Pr \left(\left| \frac{1}{T} \sum_{i=1}^{T} \left(\frac{1}{N_s}\sum_{l=1}^{N_s} x_{j,l}^{(i)} \right)\left(\frac{1}{N_s}\sum_{l=1}^{N_s} x_{k,l}^{(i)} \right)\right.\right. \left.\left.- \left(\vb{A}^{N_s}\right)_{j,k}\right| \geq\frac{\varepsilon}{6\gamma^2}\right)\leq \frac{\delta}{3N^2},\\
    &\Pr \left(\left| \frac{1}{T} \sum_{i=1}^{T} {\left(\frac{1}{N_s}\sum_{l=1}^{N_s} x_{j,l}^{(i)} \right)} y^{(i)} -\left(\vb{b}\right)_{j}\right| \geq\frac{\varepsilon}{12\gamma}\right)\leq \frac{\delta}{3N},\\
    &\Pr \left(\left| \frac{1}{T} \sum_{i=1}^{T} \left(y^{(i)}\right)^2-\Et y^2\right| \geq\frac{\varepsilon}{6}\right)\leq \frac{\delta}{3},
    \end{align*}
  provided $T\geq \ln \left(\frac{6N^2}{\delta}\right)
   \frac{{\left(12\gamma^2\|O\|^2\right)}^2}{\varepsilon^2} \in \mathcal{O}\left(\ln \left(\frac{N}{\delta}\right)\frac{1}{\varepsilon^2}\right)$. By the union bound, the inequalities
   \begin{align*}
    &\left\|\vb{A}_T^{N_s}-\vb{A}^{N_s}\right\|_{\text{max}} \leq \frac{\varepsilon}{6\gamma^2}\\
    &\left\|\vb{b}^{N_s}_T - \vb{b}\right\|_\infty \leq \frac{\varepsilon}{12\gamma}\\
    & \left|Y_T - Y\right|\leq \frac{\varepsilon}{6}
   \end{align*}
    hold at the same time with probability at least $1-\delta$. Then we have that

    \begin{align*}
      & \left|\vb{c}^{\top}\vb{A}_T^{N_s}\vb{c} - 2 \left(\vb{b}^{N_s}_T\right)^\top \vb{c} + Y_T - \left(\vb{c}^{\top}\vb{A}^{N_s}\vb{c} - 2 \vb{b}^\top \vb{c} + Y\right)\right|\\
      \leq & \gamma^2 \left\|\vb{A}_T^{N_s}-\vb{A}^{N_s}\right\|_{\text{max}} + 2\gamma \left\|\vb{b}^{N_s}_T - \vb{b}\right\|_\infty + \left|Y_T - Y\right|\\
      \leq & \gamma^2 \frac{\varepsilon}{6\gamma^2}+ 2\gamma \frac{\varepsilon}{12\gamma}+\frac{\varepsilon}{6}\\
      \leq & \varepsilon/2
    \end{align*}
    holds with probability at least $1-\delta$, which concludes the proof.
\end{proof}

Based on \cref{lemma:relation_MSE}, we next prove the following theorem.
\begin{theorem}\label{thmapp:convergence_MSE}
  A training set
  \begin{equation*}
    \mathbb{S} = {\left\{\left( \left(x_1^{(i)}, \ldots, x_N^{(i)}\right),
          y^{(i)} \right)\right\}}_{i=1,\ldots,T}
  \end{equation*}
  of size
  \begin{equation*}
    T\geq \ln \left(\frac{6N^2}{\delta}\right)
    \frac{{\left(12\gamma^2\|O\|^2\right)}^2}{\varepsilon^2} \in \mathcal{O}\left(\ln \left(\frac{N}{\delta}\right)\frac{1}{\varepsilon^2}\right)
  \end{equation*}
  suffices to obtain a solution $\vb{c}_T^{N_s}$ such that
  \begin{equation*}
    L(\vb{c}_T^{N_s}) \leq L(\vb{c}_*^{N_s}) + \varepsilon
  \end{equation*}
  with probability at least $1-\delta$.
\end{theorem}

\begin{proof}
  According to \cref{lemma:relation_MSE}, given
  $T\geq \ln \left(\frac{6N^2}{\delta}\right) \frac{{\left(12\gamma^2\|O\|^2\right)}^2}{\varepsilon^2} \in \mathcal{O}\left(\ln \left(\frac{N}{\delta}\right)\frac{1}{\varepsilon^2}\right)$,
  we have that
  \begin{align*}
    &\vb{c}_T^{N_s {\top}} \vb{A}_T^{N_s}  \vb{c}_T^{N_s} - 2 \left(\vb{b}^{N_s}_T\right)^\top \vb{c}_T^{N_s} + Y_T \\
    = &\min_{\|\vb{c}\|_1 \leq \gamma} \vb{c}^{\top}\vb{A}_T^{N_s}\vb{c} - 2 \left(\vb{b}^{N_s}_T\right)^\top \vb{c} + Y_T  \\
    \leq &\vb{c}_*^{N_s {\top}}\vb{A}_T^{N_s}\vb{c}_*^{N_s {\top}} - 2 \left(\vb{b}^{N_s}_T\right)^\top \vb{c}_*^{N_s {\top}} + Y_T \\
    \leq & \vb{c}_*^{N_s {\top}}\vb{A}^{N_s}  \vb{c}_*^{N_s} - 2 \vb{b}^\top  \vb{c}_*^{N_s} + Y +\varepsilon/2
  \end{align*}
  with probability at least $1-\delta$. Then we have that
  \begin{align*}
    &L(\vb{c}_T^{N_s}) - L(\vb{c}_*^{N_s})\\
    = & \vb{c}_T^{N_s {\top}} \vb{A}^{N_s}  \vb{c}_T^{N_s} - 2 \vb{b}^\top  \vb{c}_T^{N_s} +Y - \left(\vb{c}_*^{N_s {\top}} \vb{A}^{N_s}  \vb{c}_*^{N_s} - 2 \vb{b}^\top  \vb{c}_*^{N_s} +Y\right)\\
    \leq & \vb{c}_T^{N_s {\top}} \vb{A}^{N_s}  \vb{c}_T^{N_s} - 2\vb{b}^\top  \vb{c}_T^{N_s} +Y - \left(\vb{c}_T^{N_s {\top}} \vb{A}_T^{N_s}  \vb{c}_T^{N_s} - 2 \left(\vb{b}^{N_s}_T\right)^\top \vb{c}_T^{N_s} + Y_T\right) + \varepsilon/2.
  \end{align*}
  Since $\|\vb{c}_T^{N_s}\|_1 \leq \gamma$, according to \cref{lemma:relation_MSE}, it holds that
  \begin{align*}
    &\vb{c}_T^{N_s {\top}} \vb{A}^{N_s}  \vb{c}_T^{N_s} - 2 \vb{b}^\top  \vb{c}_T^{N_s} +Y - \left(\vb{c}_T^{N_s {\top}} \vb{A}_T^{N_s}  \vb{c}_T^{N_s} - 2 \left(\vb{b}^{N_s}_T\right)^\top \vb{c}_T^{N_s} + Y_T\right)\leq \varepsilon/2,
  \end{align*}
  which implies that $L(\vb{c}_T^{N_s}) - L(\vb{c}_*^{N_s}) \leq \varepsilon$.
 \end{proof}

\subsection{{Time Complexity of NIL}}\label{subapp:time_complexity}

We analyze the time complexity of our method.
Suppose the number of neighbor circuits is $N$, and we want the test MSE of our
model to be smaller than $\varepsilon$ with probability at least $1-\delta$.
We estimate the cost of our approach.
We first estimate the number of shots $N_s$ needed to measure the involved
circuits.
Note that the MSE derived from the optimal coefficients $\vb{c}_*^{N_s}$ is
given by
\[
  \vb{c}_*^{N_s \top} \Et \left[A +\ReNs\right]\vb{c}_*^{N_s} -2\vb{b}^{\top} \vb{c}_*^{N_s} + \Et[y^2].
\]  %
Since $X_i \leq \|O^2\|$ and $\|\vb{c}_*^{N_s}\|_1 \leq \gamma$, the
regularization term introduced by shot noise is upper bounded by
\begin{align*}
  \vb{c}_*^{N_s \top} \Et \left[\ReNs\right]\vb{c}_*^{N_s} = & \vb{c}_*^{N_s \top} \Et \left[\frac{\diag(X_1,\ldots, X_N)}{N_s}\right]\vb{c}_*^{N_s}\\
   \leq & \|O^2\|\gamma^2/N_s.
\end{align*}

To ensure that this term does not dominate the error in the MSE, we impose the
condition $\|O^2\|\gamma^2/N_s \leq \varepsilon$, which gives that
$ N_s \geq \|O^2\|\gamma^2/\varepsilon$.
{It is worth noting that the above expression may be rather conservative and
  far from the tight condition on the number of shots required to achieve an MSE
  of $\varepsilon$.
  In the numerical experiments, we observe
  that using $N_s$ shots to estimate the expectation value can yield an MSE that
  is significantly smaller than $\|O^2\|\gamma^2 / N_s$.}

In the following, we treat $N_s$ as a constant satisfying the above condition.
The time complexity of our approach can be estimated through the following
steps: generating the training set, computing the optimal coefficients given in
\cref{eqapp:lasso_regression}, and applying the learned coefficients to mitigate
noise in the target circuit.
Below, we discuss the time complexity of each step individually.

\paragraph{Generating the training set.} To generate a training set of the form
\begin{equation*}
  \mathbb{S}^{N_s}  = {\left\{\left(\left(\frac{1}{N_s}\sum_{j=1}^{N_s} x_{1,j}^{(i)}, \ldots, \frac{1}{N_s}\sum_{j=1}^{N_s} x_{N,j}^{(i)}\right),
        y^{(i)}\right)\right\}}_{i\in [T]},
\end{equation*}%
we need to compute the feature vector
$\left(\frac{1}{N_s}\sum_{j=1}^{N_s} x_{1,j}^{(i)}, \ldots, \frac{1}{N_s}\sum_{j=1}^{N_s} x_{N,j}^{(i)}\right)$
and the label $y^{(i)}$ for each training circuit generated by
2-training method.
The number of training circuits we generate is
$T = \mathcal{O}\left(\ln \left(\frac{{N}}{\delta}\right)\frac{{1}}{\varepsilon^2}\right)$.
This choice can be justified as follows.
As proved in \cref{subapp:estimation_trainingsize}, a training set of size
$T = \mathcal{O}\left(\ln \left(\frac{N}{\delta}\right)\frac{1}{\varepsilon^2}\right)$
suffices to obtain a solution $\vb{c}_{t}^{N_s}$ whose test MSE is smaller than
the optimal test MSE plus $\varepsilon$.
If the optimal test MSE is already smaller than $\varepsilon$, then the test MSE of
$\vb{c}_{t}^{N_s}$ will be below $2\varepsilon$, which is acceptable.
However, if the test MSE of $\vb{c}_{t}^{N_s}$ exceeds $2\varepsilon$, the optimal
MSE must be larger than $\varepsilon$.
In this case, we should increase the number of neighbor circuits rather than the
number of training circuits.

After generating the training circuits, we apply the neighbor map to these
circuits and measure
$\mathcal{O}\left(T N\right) = \mathcal{O}\left(\ln \left(\frac{{N}}{\delta}\right)\frac{{N}}{\varepsilon^2}\right)$
neighbor circuits (with $N$ neighbor circuits for each training circuit) for
$N_s$ times.
The expectation values for the ideal training circuits can be simulated
efficiently, as they are all Clifford circuits.
The labels $y^{(i)}$'s can be efficiently computed classically in
$\mathcal{O}\left(T \operatorname{poly}(n)\right)$ time, where $n$ is the number
of qubits.

In our numerical experiments, we observe that the training set size required for
convergence is significantly smaller than the theoretical bound.
Specifically, empirical data fitting suggests that the training MSE converges when the training set size follows the approximate scaling
\[
T = \frac{2\gamma \ln N}{\sqrt{\varepsilon}},
\]
where $\gamma$ is some constant.
Based on this empirical scaling, it appears sufficient in practice to generate approximately $\order{TN} = \order{\frac{N \ln N}{\sqrt{\varepsilon}}}$ neighbor circuits.

\paragraph{Computing the optimal coefficients.}
This step can be completed on classical computers.
Many algorithms have been developed to solve the optimization problem
\cref{eqapp:lasso_regression}~\cite{schmidt2005least,mairal2012complexity,roth2004generalized,hazan2012linear}.
Particularly, we can solve this problem with the algorithm
in~\cite{schmidt2005least} in time
$\mathcal{O}\left(N^3+T N^2\right) = \mathcal{O}\left(N^3\right)$.

The above two steps constitute the entire training stage, where we generate
$\mathcal{O}\left(T N\right) = \mathcal{O}\large(\ln \left(\frac{{N}}{\delta}\right)\\ \frac{{N}}{\varepsilon^2}\large)$
quantum circuits, and the running time on a classical computer is upper bounded
by
\[
\mathcal{O}\left(T (\operatorname{poly}(n))+N^3\right) = \mathcal{O}\left(\ln \left(\frac{{N}}{\delta}\right)\frac{{\operatorname{poly}(n)}}{\varepsilon^2}+N^3\right).
\]
The coefficients returned by the classical computer for QEM are applicable to
all target circuits with the same structure, i.e., circuits that differ only in
the parameters of the rotation gates.
Next, we utilize the learned coefficients to perform QEM on the target circuit.

\paragraph{Applying the learned model.}
After obtaining the coefficients $\vb{c}_{t}^{N_s}$, for a target circuit
requiring QEM, we only need to measure its $N$ neighbor circuits $N_s$ times to
collect the feature vector $\vb{x}_{N_s}$.
Finally, we output the mitigated target expectation value as
$\vb{c}_{t}^{N_s \top}\vb{x}_{N_s}$, which can be computed in $\mathcal{O}(N)$
time.

As we can see, once the coefficients are learned, applying our method to the
target circuit is very efficient, requiring significantly less time compared to
the training phase.
We emphasize that the computational resources consumed in the training stage are
worthwhile, as the learned coefficients can be applied to all parameter
configurations of the target circuit.

\section{Benchmarking Performance on Large-scale Non-Clifford Circuits}\label{app:same_mse}


In this section, we demonstrate that the performance of NIL can be reliably benchmarked even on large non-Clifford circuits, a regime that poses significant challenges for previous learning-based QEM protocols.
Owing to the specific mathematical structure of the 2-design training method, we can establish the following corollary:

\begin{Corollary}\label{cor:same_mse}
  For an arbitrary $\vb{c} = {\left(c_1, c_2, \cdots c_N\right)}^{\top}$,
  \begin{equation*}
    \Etrain {\left(\vb{c}^{\top} \cdot \vb{x}_{N_s}-y\right)}^2 =
    \Etest {\left(\vb{c}^{\top} \cdot \vb{x}_{N_s}-y\right)}^2.
 \end{equation*}
\end{Corollary}

\begin{proof}
  We directly substitute $\vb{c}$ to the target circuit.
  According to \cref{thm:optimal_shot_noise}, we have
  \begin{align*}
   & \Etest {\left(\vb{c}^{\top} \cdot \vb{x}_{N_s}-y\right)}^2 \\
   = & \Etest {\left( \vb{c}^{\top} \cdot \vb{x}_{N_s}-y\right) \left(\vb{c}^{\top} \cdot \vb{x}_{N_s}-y\right)}^{\top} \\
   = & \vb{c}^{\top} \Etest \left[\vb{x}_{N_s}\vb{x}_{N_s}^{\top} \right] \vb{c} -2 \vb{c}^{\top}\Etest\left[y\vb{x}_{N_s}\right]  + \Et y^2\\
   = &  \vb{c}^{\top} \Etrain \left[\vb{x}_{N_s}\vb{x}_{N_s}^{\top} \right] \vb{c} -2 \vb{c}^{\top}\Etrain\left[y\vb{x}_{N_s}\right]  + \Et y^2\\
   = & \Etrain {\left(\vb{c}^{\top} \cdot \vb{x}_{N_s}-y\right)}^2.
  \end{align*}
\end{proof}

\Cref{cor:same_mse} reveals an important fact: when a learning-based QEM protocol
employs linear regressions to fit the training data, the set of training
circuits generated by 2-training method can also efficiently benchmark the
obtained model's actual performance on the target parameterized circuit in the
average sense, provided the neighbor circuits of both training and target
circuits are measured with the same number of shots.

To understand this, note that the MSE achieved on the training circuits
generated by 2-training method can be computed classically and
efficiently, whereas the MSE for target parameterized circuits is difficult to
obtain even with quantum computers, since ideal expectation values cannot be
efficiently simulated classically.
Therefore, \cref{cor:same_mse} enables us to estimate our protocol's performance
on large non-Clifford quantum circuits on average without requiring precise
classical simulation of the ideal target expectation values.

\section{Mean Squared Error and Worst Case Error}\label{app:mse_wce}

We have shown that the learning-based QEM model using our approach to generate
training circuits yields the optimal linear function (with respect to MSE) on
average.
In this section, we prove that if the MSE on the training dataset is low, then for any
specific parameter configuration of the target quantum circuit, there is a high probability of obtaining good error mitigation performance.

Let the MSE achieved on the training circuit be $\varepsilon$, i.e.,
\[
  \Etrain {\left(\vb{c}^{\top} \cdot \vb{x}_{N_s}-y\right)}^2 = \varepsilon .
\]
According to \cref{cor:same_mse}, we have
\begin{equation*}
  \Etest {\left(\vb{c}^{\top} \cdot \vb{x}_{N_s}-y\right)}^2 =
  \Etrain {\left(\vb{c}^{\top} \cdot \vb{x}_{N_s}-y\right)}^2 = \varepsilon.
\end{equation*}
By Jensen's inequality, we derive that
\begin{equation*}
  {\left(\Etest|\vb{c}^{\top} \cdot \vb{x}_{N_s}-y|\right)}^2 \leq
  \Etest {\left(\vb{c}^{\top} \cdot \vb{x}_{N_s}-y\right)}^2 = \varepsilon,
\end{equation*}
which yields
\[
  \Etest|\vb{c}^{\top} \cdot \vb{x}_{N_s}-y| \leq \sqrt{\varepsilon}.
\]
Therefore, for any instance of $\left(\vb{x}_{N_s}', y'\right)$, which is
obtained by measuring each neighbor circuit of the target circuit for $N_s$
times, by Chebyshev's inequality it holds that
\begin{align*}
  \Pr \left(\left||\vb{c}^{\top} \cdot \vb{x}_{N_s}'-y'| -
  \Etest|\vb{c}^{\top} \cdot \vb{x}_{N_s}-y|\right| \geq
  k\sqrt{\varepsilon}\right)  \leq \frac{\Var|\vb{c}^{\top} \cdot \vb{x}_{N_s}-y|}{k^2\varepsilon}\leq \frac{1}{k^2}.
\end{align*}
This suggests that on any specific parameter configuration of the target quantum circuit,
the protocol can achieve
\begin{equation*}
  |\vb{c}^{\top} \cdot \vb{x}_{N_s}'-y'|\leq k\sqrt{\varepsilon} +
  \sqrt{\varepsilon} = (k+1)\sqrt{\varepsilon}
\end{equation*}
with probability at least $1-\frac{1}{k^2}$.

For instance, if the MSE on the training set is less than $10^{-4}$ (which is
easily achievable for all the quantum circuits we have numerically studied) and
we choose $k = 10$, then our approach ensures that on any specific parameter
configuration of the target circuit, the error in the output expectation value
will be less than 0.1 with probability exceeding 99\%.
Note that we do not rule out the possibility of encountering certain parameter
configurations with larger errors.
However, the above discussion implies that their occurrence is statistically
improbable.
In fact, we have conducted various numerical experiments on different quantum
circuits to investigate the QEM performance for individual circuits.
In all cases we have studied (as detailed in \cref{subapp:mse_wce_numerical}), when
a good QEM protocol is discovered in terms of MSE, the expectation values are
recovered to be very close to the ideal values for all test circuits.

\subsection{Low Mean Squared Error May Indicate Low Worst-Case Error}\label{subapp:mse_wce_numerical}
\begin{figure}[H]
  \captionsetup[subfigure]{justification=centering, labelfont = bf}
  \begin{subfigure}[t]{0.33\textwidth}
    \centering
    \includegraphics[width = \linewidth]{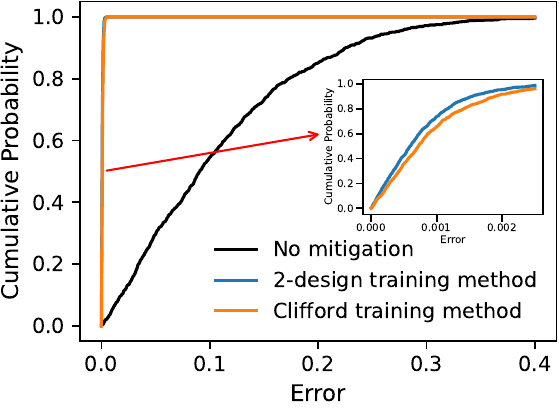}
    \subcaption{$\mathsf{vqe}$-6-4 {(6 qubits, 13 layers)}}
  \end{subfigure}%
  \hfill
  \begin{subfigure}[t]{0.33\textwidth}
    \centering
    \includegraphics[width = \linewidth]{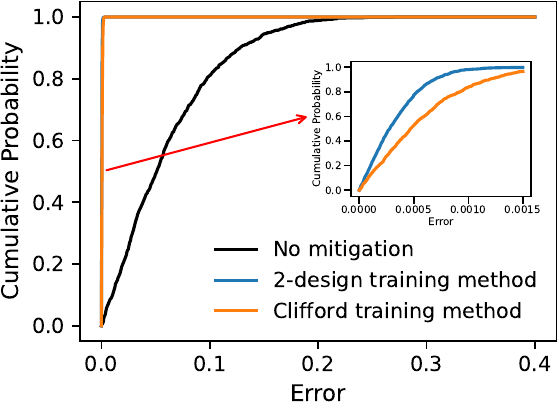}
    \subcaption{$\mathsf{vqe}$-$R_y$-$6$-$4$ {(8 qubits, 17 layers)}}
  \end{subfigure}%
  \hfill
  \begin{subfigure}[t]{0.33\textwidth}
    \centering
    \includegraphics[width = \linewidth]{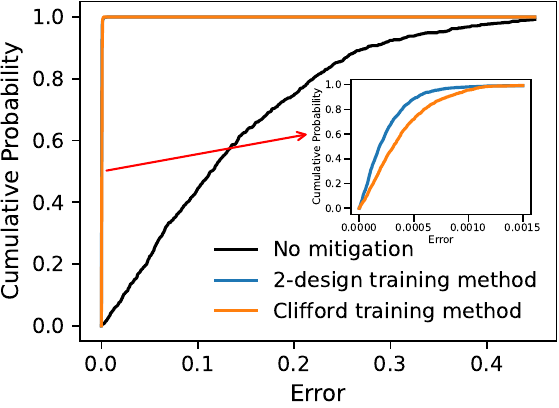}
    \subcaption{$\mathsf{vqe}$-$(3,2)$-$2$ {(8 qubits, 20 layers)}}
  \end{subfigure}%
  \hfill

  \vspace{\baselineskip}

  \begin{subfigure}[t]{0.33\textwidth}
    \centering
    \includegraphics[width = \linewidth]{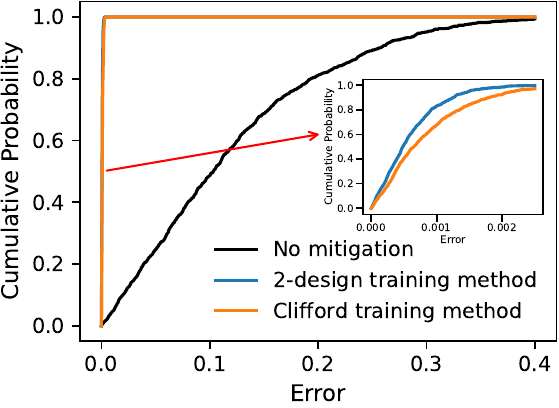}
    \subcaption{$\mathsf{vqe}$-$8$-$4$ {(6 qubits, 13 layers)}}
  \end{subfigure}%
  \hfill
  \begin{subfigure}[t]{0.33\textwidth}
    \centering
    \includegraphics[width = \linewidth]{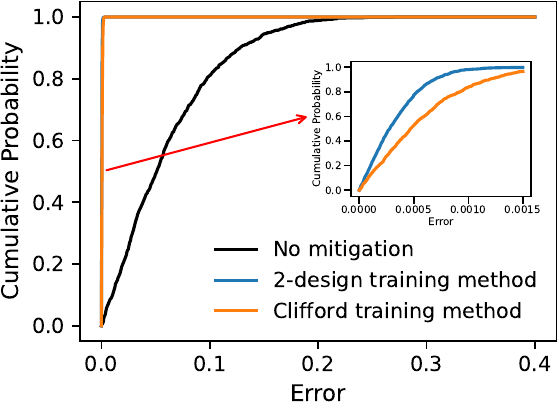}
    \subcaption{$\mathsf{vqe}$-$R_y$-$8$-$4$ {(8 qubits, 17 layers)}}
  \end{subfigure}%
  \hfill
  \begin{subfigure}[t]{0.33\textwidth}
    \centering
    \includegraphics[width = \linewidth]{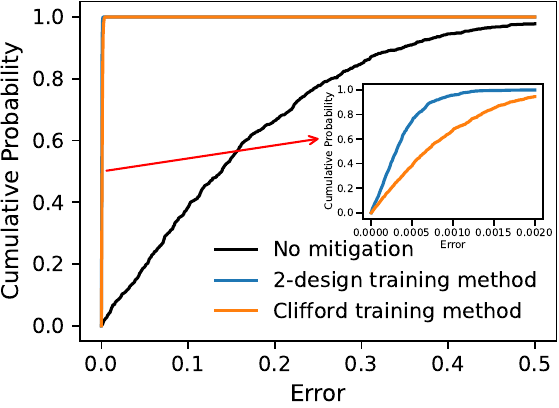}
    \subcaption{$\mathsf{vqe}$-$(4,2)$-$2$ {(8 qubits, 20 layers)}}
  \end{subfigure}%

  \caption{Empirical cumulative distribution of the computing error for 500
    random parameter configurations of the target quantum circuit.
    For each point $(x,y)$ on the curve, $x$ represents the absolute deviation
    from the ideal expectation value, and $y$ denotes the proportion of
    configurations (out of 500) with an absolute deviation less than $x$.
    In the main figures, the blue and orange curves nearly coincide.
    The insets zoom in on the region where the curves diverge, specifically when
    the error approaches zero.
    The results show that the blue curve (obtained using the 2-design method)
    yields a more accurate prediction of the ideal expectation value with higher
    probability compared to the orange curve (obtained using the Clifford
    training method).}%
  \label{fig:cumulative_distribution}
\end{figure}
We perform numerical simulations for a specific error mitigation task to evaluate the performance of our protocol on individual parameter configurations of the target quantum circuit.
Recall that our protocol guarantees the optimal linear solution in the average
case.

Here all the weight-1 Pauli neighbors are chosen to serve as neighbor circuits,
as depicted in \cref{fig:cumulative_distribution}.
We find that for all the parameter configurations where we could classically
calculate the ideal expectation values, the worst-case error of our protocol is
less than $3\times 10^{-3}$.
This indicates that our method not only has excellent error mitigation effects
on average, but also has a low failure probability.

We also compare the performance of our method with the Clifford training method,
which uses all Clifford gates.
As shown in the insets of \cref{fig:cumulative_distribution}, our method
achieves lower error with higher probability across all cases.

\section{Using Non-Clifford Quantum Circuits for Training}\label{app:mixed_non_clf}

In this section, we show that retaining a subset of non-Clifford gates when generating the training circuits helps the model learn more efficiently.
We first show a limitation of Clifford training circuits in learning the pattern
of quantum noise.
When considering only Pauli noise channels and a Pauli observable $O$, if a
neighbor circuit $C^{(i)}_j$ is obtained by inserting Pauli gates to a Clifford
training circuit $C^{(i)}$, then it can be proven that
$\bignavg{\C^{(i)}_{j}} = \pm \navg{\C^{(i)}}$, where $\navg{\C}$ is the noisy
target expectation value obtained on $\C$.
\begin{lemma}\label{lemma:form_noisy_clf}
  Suppose a Clifford quantum circuit $C_j$ can be expressed as
  $C_j = U_m \cdots U_1$, and a noisy implementation of $C_j$ is given by
  \begin{equation*}
    \noise{\C}_{j} = \E_m \circ \P_{m} \circ \U_m \cdots \E_1 \circ \P_{1} \circ \U_1,
  \end{equation*}
  where $\E_j$ is a Pauli noise channel, $\U_j$ is a Clifford gate, and $\P_j$
  is a Pauli gate or the product of two Pauli gates.
  Then $\bignavg{\C^{(i)}_{j}} = \pm \navg{\C^{(i)}}$.
\end{lemma}

\begin{proof}
  Since $\E_m$ is a Pauli noise channel, we have that
  $\E_m \circ \P_m = \P_m \circ \E_m$.
  In this way, $\P_m$ can be moved to the end of the circuit.
  Similarly, $\P_{m-1}$ commutes with $\E_{m-1}$.
  As $U_m$ is a Clifford gate; that is,
  $U_m P_{m-1} \rho P_{m-1} U_m^{\dagger} = P'_{m-1} U_m \rho U_m^{\dagger} P'_{m-1}$,
  where $P'_{m-1}$ is also a Pauli operator.
  Then the quantum channel $\P_{m-1}'$ commutes with $\E_m$ and can also be
  moved to the end of the circuit, which combined with $\P_m$ becomes another
  Pauli operator.

  Repeating this process, we move all the inserted Pauli gates to the end of the
  circuit, resulting in
  \begin{equation*}
    \noise{\C}_{j} = \Q \circ \E_m \circ \U_m \cdots \E_1 \circ \U_1,
  \end{equation*}
  where $\Q$ is a Pauli channel.
  Therefore, we have that
  \begin{align*}
    \navg{\C_{j}}
    & = \tr \bigl(O \, \noise{\C}_{j} (\rho_{\lin}) \bigr) \\
    & = \tr(O \, \Q \circ \E_m \circ \U_m \cdots \E_1 \circ \U_1 (\rho_{\lin})) \\
    & = \tr( (Q^{\dagger} O Q) \, \E_m \circ \U_m \cdots \E_1 \circ \U_1 (\rho_{\lin})) \\
    & = \mu(Q, O) \tr \bigl( O \, \noise{\C} (\rho_{\lin}) \bigr) \\
    & = \mu(Q, O) \navg{\C},
 \end{align*}
 where $\mu(Q, O)$ is $1$ if $Q$ and $O$ commute and $-1$ otherwise.
\end{proof}

\begin{algorithm}[H]
  \textbf{Input:} Target parameterized circuit $C(\vb*{\theta})$ and
  the depth $L$ for the non-Clifford layer\\
  \textbf{Output:} A set $\mathbb{S}$ of training circuits
  \begin{algorithmic}[1]
    \State $S \gets \emptyset$
    \For{$i = 1, 2, \ldots, T$}
      \For{each gate in the circuit $C(\vb*{\theta})$}
        \If {the gate is of the form $R_{\sigma}(\theta)$ ($\sigma = x,y,z$)}
          \If {the current layer \(l>L\),}
            \State $\theta \gets_{\$} \{0, \pi/2, \pi, 3\pi/2\}$
          \Else
            \State $\theta \gets_{\$} [0, 2\pi)$
          \EndIf
        \EndIf
      \EndFor
      \State Add the new circuit to $\mathbb{S}$
    \EndFor
  \end{algorithmic}
  \caption{Strategy for constructing training circuits including non-Clifford
    ones}%
  \label{alg:gen_mixed_training}
\end{algorithm}

\Cref{lemma:form_noisy_clf} implies that the noisy expectation value of each
training circuit $C^{(i)}$ and those of all its neighbor circuits $C^{(i)}_j$
are either the same or the opposite in this case.
Consequently, all data points in the training set are of the form
$\left((x^{(i)}, \pm x^{(i)}, \pm x^{(i)}, \cdots), y^{(i)}\right)$ with
$x^{(i)} = \navg{\C^{(i)}}$ and the noiseless expectation value
$y^{(i)} \in \{0, \pm 1\}$~\cite{qin2023error}.
This creates a significant discrepancy between the training set and the
test set generated from the target quantum circuits, which could seriously hinder
the learning process.

A possible way to address this issue is to introduce training circuits of the
form $C^{(i)} = U^{(i)}V^{(i)}$, where $U^{(i)}$ is a Clifford circuit and
$V^{(i)}$ is a shallow quantum circuit that may contain non-Clifford gates.
To achieve this, we can modify 2-training method by setting the angles of
the rotation gates in the first few layers of $C^{(i)}$ to be uniformly
distributed in $[0, 2\pi)$, which forms $V^{(i)}$, as illustrated in
\cref{fig:gen_mixed_training_circuit}.
The pseudo-code for this strategy is presented in \cref{alg:gen_mixed_training}.

\begin{figure}[htb]
  \centering
  \includegraphics[width = 0.3\columnwidth]{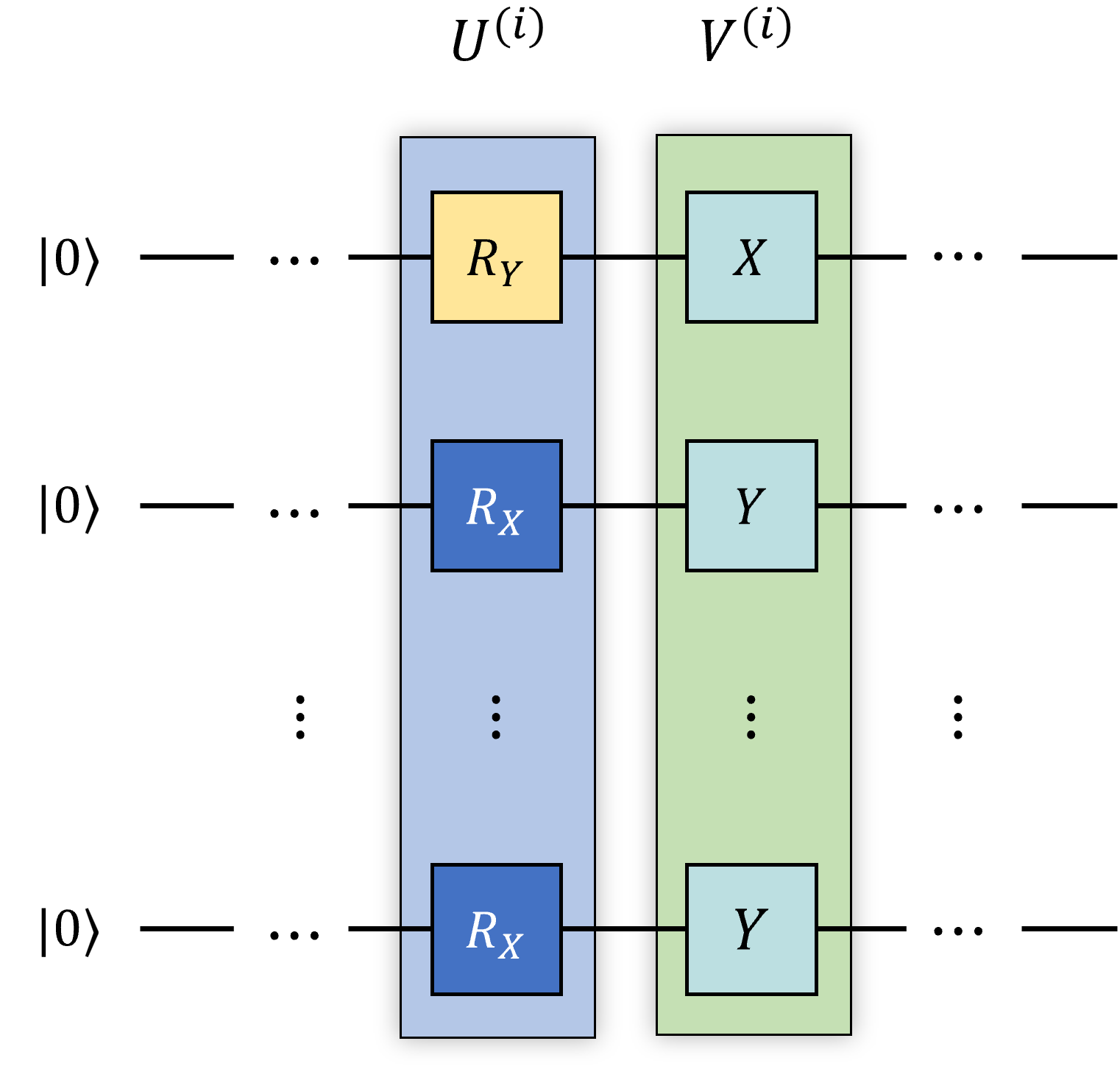}
  \caption{Non-Clifford training circuits.
    It is produced by first applying 2-training method on the target
    circuit, and then replacing the first layer by non-Clifford gates.}%
  \label{fig:gen_mixed_training_circuit}
\end{figure}

When the observable $O$ is a Pauli operator or a sum of polynomially many Pauli
operators, the expectation value $\avg{\C^{(i)}}$ can be classically computed.
Indeed, suppose for simplicity that $O$ is a Pauli operator, then the ideal
expectation value is given by
\begin{equation*}
  \avg{\C^{(i)}} = \tr (O\, \U^{(i)} \circ \V^{(i)} \bigl(\rho_{\lin}\bigr))
  = \tr(O'\, \V^{(i)} \bigl(\rho_{\lin}\bigr)),
\end{equation*}
where $O' := {\bigl(U^{(i)}\bigr)}^{\dagger} O U^{(i)}$ is also a Pauli operator
by the definition of Clifford circuits.
In such a situation, computing the expectation value for $C^{(i)}$ is reduced to
computing that for $V^{(i)}$, which can then be classically computed by tensor
network methods~\cite{markov_simulating_2008}.
For the computation of $\avg{\C^{(i)}}$ to be efficient, we require the depth of
$V^{(i)}$ to be $O(\log{n})$, where $n$ is the number of input qubits.

\subsection{Comparison Between Mixed and Purely Clifford Training Circuits}

Next we compare the performance of a learning-based QME protocol
employing shallow non-Clifford circuits for training with that of a protocol
employing only Clifford training circuits.
Specifically, in the first one or two layers of each training circuit, we allow
non-Clifford gates (i.e., rotation gates with random angles), while the
remaining layers are still generated by 2-training method.
Here, we choose all the weight-1 neighbors as the neighbor circuits, and the
learning model remains a linear function.
It can be seen that both protocols successfully converge to the optimal linear
function.
We then compare the convergence speeds of these two protocols to the solutions,
particularly the sizes of the training circuit sets required to achieve stable
convergence.

As shown in \cref{fig:non-Clifford_vs_Clifford(lasso)}, the protocol using
non-Clifford training circuits exhibits a smoother convergence curve, indicating
a faster convergence rate.
Therefore, in practical applications, we can reduce the number of training
circuits required by introducing shallow non-Clifford gate layers into the
training circuits.

\begin{figure}[htb]
  \centering
  \includegraphics[width = 0.4\columnwidth]{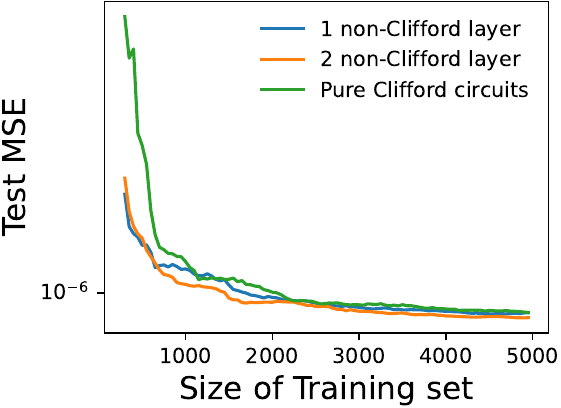}
  \caption{Comparison between the performance of learning-based QEM models with
    purely Clifford training circuits and that with non-Clifford training
    circuits.
    The underlying circuit is $\mathsf{vqe}$-6-4.}%
  \label{fig:non-Clifford_vs_Clifford(lasso)}
\end{figure}

\section{Producing Training Circuits with All the Single-Qubit Clifford Gates}\label{app:solution_all_clifford}

In this section, we show that the Clifford training method performs well only when the single-qubit gates in the target circuits are drawn from the Haar random distribution.
This discrepancy arises because Clifford gates do not form a quantum rotation
$2$-design, but rather a unitary $2$-design with respect to Haar random
unitaries.
Suppose the training set
\begin{equation*}
  \mathbb{S}' = {\left\{\left(\left(x_1^{(i)}, \ldots x_N^{(i)}\right),
      y^{(i)}\right)\right\}}_{i=1,\ldots,|\mathbb{S}'|}
\end{equation*}
is generated by Clifford training method, i.e., all the single-qubit
Clifford gates are utilized to replace each non-Clifford single-qubit gates in
target circuits.
Let $P_{\mathbb{S}}'$ be the probability distribution of the data in $S'$, when
the parameters of the target quantum circuit are uniformly distributed within
$[0,2\pi]$.
As in \cref{app:prove_optimal}, we can derive that the solution to the linear
regression that fits the data in $S'$ is expressed as
\begin{align}\label{eq:solution_train_all_clifford}
  {\left(\underset{P_{\mathbb{S'}}}{\Exp}\left[\vb{A}\right]\right)}^{-1}
  \underset{P_{\mathbb{S'}}}{\Exp}\left[\vb{a}\right],
\end{align}
where we write $\underset{\left(\vb{x}, y\right) \sim P_{\mathbb{S}'}}{\Exp}$ as
$\underset{P_{\mathbb{S}'}}{\Exp}$ for simplicity and the definitions of $\vb{A}$ and
$\vb{a}$ can be found in \cref{eq:coefficient_matrix}.
If $\underset{P_{\mathbb{S}'}}{\Exp} \left[\vb{x} \cdot \vb{x}^{\top}\right]$ is
not invertible, then we can use the pseudoinverse.
Repeating the proof in \cref{lemma:two_design}, and using the fact that the
Clifford group forms a unitary 3-design~\cite{webb2015clifford}, we have
\begin{align*}
  & \underset{g\in C_1}{\Exp} \tr
    \left(g \otimes I \rho g^{\dagger} \otimes I \Phi_2^{\dagger}(O)\right)
    \tr\left(g \otimes I \rho^{\prime} g^{\dagger}
    \otimes I \Phi_2^{\prime\dagger}(O)\right) \\
  = &\underset{U_1\sim\text{haar}}{\Exp}
      \tr\left(U_1 \otimes I \rho U_1^{\dagger}
      \otimes I \Phi_2^{\dagger}(O)\right)
      \tr\left(U_1 \otimes I \rho^{\prime} U_1^{\dagger}
      \otimes I \Phi_2^{\prime\dagger}(O)\right).
\end{align*}
By a similar proof as in \cref{lemma:two_design}, we can directly derive the following conclusion.

\begin{lemma}\label{lemma:two_design_all_clifford}
  Consider a new set of test circuits, where the circuits have a same structure
  with the original test circuit, with the difference being that the
  non-Clifford gates are not fixed-axis rotation gates, but Haar random
  single-qubit gates.
  Then we have
  \begin{equation*}
    \underset{P_{\mathbb{S'}}}{\Exp}\left[\vb{A}\right] =
    \underset{P_{\mathbb{\rm test'}}}{\Exp}\left[\vb{A}\right] \text{ and }
    \underset{P_{\mathbb{S'}}}{\Exp}\left[\vb{a}\right] =
    \underset{P_{\mathbb{\rm test'}}}{\Exp}\left[\vb{a}\right],
 \end{equation*}
 where $P_{\rm test'}$ is the probability distribution of the data produced by
 the new test set.
\end{lemma}

Based on \cref{lemma:two_design_all_clifford}, the solution in
\cref{eq:solution_train_all_clifford} equals
\begin{align}\label{eq:solution_test_all_clifford}
  {\left(\underset{P_{\mathbb{\rm test'}}}{\Exp}
  \left[\vb{A}\right]\right)}^{-1}
  \underset{P_{\mathbb{\rm test'}}}{\Exp}\left[\vb{a}\right],
\end{align}
which indicates that in a learning-based QEM protocol, if training circuits are
generated by replacing each single-qubit non-Clifford gate with all possible
single-qubit Clifford gates uniformly---the conventional strategy in
literature---the model will obtain the optimal linear solution on average,
provided all single-qubit non-Clifford gates in the target circuit are
Haar-random unitaries.
This also explains the significant differences observed between protocols using
different training circuit sets in the manuscript.

\section{Comparisons Between Different Design Choices of NIL-based QEM}%
\label{app:comparison_different_settings}
We now compare different design and hyperparameter choices of NIL in a comprehensive manner.

\subsection{Comparison Between Choices of Neighbor Maps}\label{subapp:comparison_neighbors}

We first compare the performance of learning-based QEM models constructed with different types of neighbor circuits. The test circuit is chosen to be $\mathsf{vqe}$-6-4, as introduced in the manuscript. Training circuits are generated using the 2-design method. For comparison, we consider three neighbor-generation strategies: Pauli-insertion neighbors, CPTP-insertion neighbors, and ZNE neighbors. The definitions of Pauli-insertion and ZNE neighbors follow those in the manuscript.

The CPTP-insertion neighbors are obtained by expanding the Pauli gate set used in Pauli-insertion neighbors to
\begin{equation*}
  \GG_1 := \left\{\X, \Y, \Z, \K^{\dagger} \S^{\dagger} \K,
    \K \S^{\dagger} \K^{\dagger}, \S^{\dagger}, \K \H \K^{\dagger}, \H,
    \K^{\dagger} \H \K \right\},
\end{equation*}
where $\S$, $\H$, and $\K$ are the quantum channels for the $S$ gate, the $H$
gate, and the $K := SH$ gate, respectively.
$\GG_1$ is actually a basis for the single-qubit CPTP maps when appended with
the identity channel and three extra state-preparation
channels~\cite{takagi_optimal_2021}. Note that all the gates from $\GG_1$ are single-qubit Clifford gates. For brevity, we hereafter refer to Pauli-insertion and CPTP-insertion neighbors simply as Pauli neighbors and CPTP neighbors, respectively.

The training method still adopts the Lasso regression formulation
in~\cref{eqapp:lasso_regression}, where we set $\gamma = 2$ for both Pauli
neighbors and CPTP neighbors.
Recall that the ZNE neighbor circuits for each training (or test) circuit are
obtained by replacing each noise channel $\Lambda$ with $\Lambda^{\alpha}$,
where $\alpha \in \{1, 1.1, 1.34, 1.58\}$.
This implies that only four neighbor circuits are available for each training
circuit.
To improve the power of this neighbor map, we set the $\ell_1$ norm restriction
$\gamma = 5$ in Lasso regression for ZNE neighbors.

\begin{figure}[H]
  \centering
  \includegraphics[width = 0.5\columnwidth]{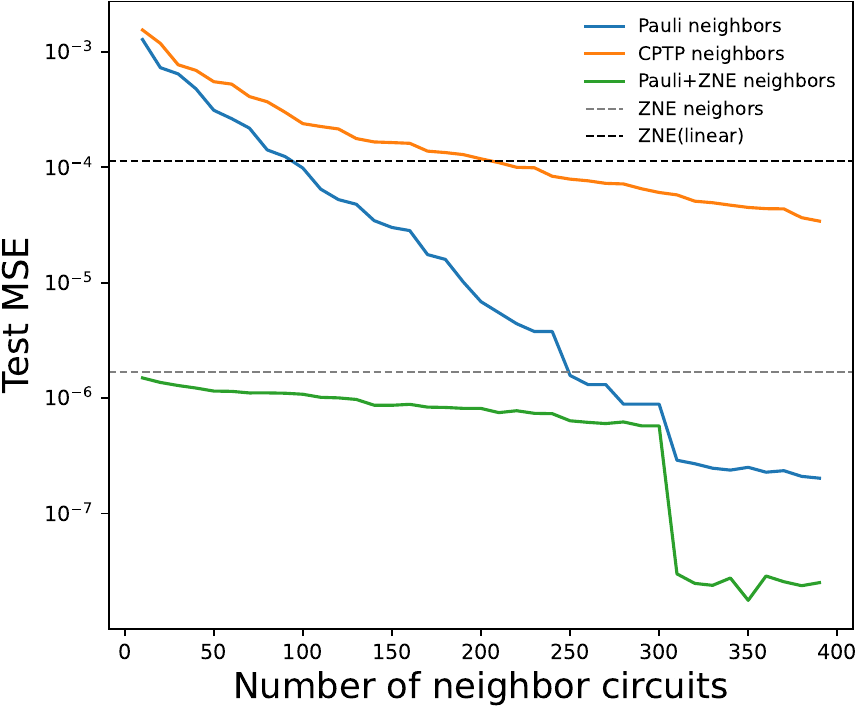}
  \caption{Performances of learning-based QEM protocols that choose different neighbor circuits.
    The underlying circuit is $\mathsf{vqe}$-$6$-$4$.}\label{fig:Performance_three_neighbor}
\end{figure}

The results are shown in~\cref{fig:Performance_three_neighbor}.
It can be observed that the performance of Pauli neighbors surpasses that of
CPTP neighbors, suggesting that Pauli neighbors provide a better choice
for PEC-inspired NIL.
Moreover, we observe that the protocol based on ZNE neighbors also achieves
competitive performance, even though it uses only four neighbor circuits per
training circuit.
It outperforms both CPTP and Pauli neighbors when the number of available
neighbor circuits is limited.
However, Pauli neighbors deliver the best performance when the number of
neighbor circuits exceeds 250 (in this case, there are a total of 300
weight-1 neighbor circuits).

To combine the power of both ZNE neighbors and Pauli neighbors, we propose a new
type of neighbor map called \emph{ZNE + Pauli neighbors}.
This $\combine$ map is constructed by combining these two approaches as
follows:
\begin{small}
\begin{align*}
  \combine \left(\!\navg{\C}_{\epsilon},\navg{\C}_{1.1\epsilon},\navg{\C}_{1.34\epsilon},\navg{\C}_{1.58\epsilon},\navg{\C_{1}}, \!\ldots\!, \navg{\C_{N}}\!\right),
\end{align*}
\end{small}%
where $\epsilon$ is the original noise rate and $C_1\ldots, C_{N}$ are the Pauli neighbors.
Again, by employing Lasso regression with $\gamma = 5$, we test this new
neighbor and observe that it performs well even with a very small number of
neighbors.
As shown in \cref{fig:Performance_three_neighbor}, this neighbor achieves a
two-order-of-magnitude improvement in MSE compared to Pauli neighbors when the
number of neighbors is small.
This further demonstrates the flexibility of our approach.

\subsection{Comparison Between Linear Models and Neural Networks}\label{subapp:comparison_learning_model}

We first show that Lasso regression can significantly reduce the computational
cost of learning-based QEM protocols compared to standard linear regression.

\Cref{tab:lambada_mse_cost} illustrates the example of $\mathsf{vqe}$-$6$-$4$.
It can be seen that by adjusting the $\ell_1$ norm constraint in
\cref{eqapp:lasso_regression}, we can significantly reduce the sampling cost of the
protocol, while the QEM performance only slightly decreases.
It is critical to note that, as evidenced by the last column of
\cref{tab:lambada_mse_cost}, the value of $\sum_i |c_i|$ obtained by ordinary
linear regression is extremely high, rendering it nearly unusable in practical
applications.
This issue stems from the ill-conditioning of the least squares coefficient
matrix.
However, implementing Lasso regression significantly reduces the value of
$\sum_i |c_i|$, thereby greatly reducing the computational cost.
Therefore, the advantage of Lasso regression deserves further investigation in
learning-based QEM protocols.
\begin{table}[H]
  \centering
  \begin{small}
  \begin{tabular}{cccc}
    \toprule
    $\gamma$ & \textbf{Training MSE} & \textbf{Test MSE} & \textbf{$\sum_i |c_i|$} \\
    \midrule
    {unconstrained}  & $8.17\times 10^{-7}$ & $7.03\times 10^{-7}$ & 11057.57  \\
    {$10$} & {$8.57\times 10^{-7}$} & {$7.93\times 10^{-7}$} & {2.6} \\
    {$2$} & {$8.87\times 10^{-7}$} & {$8.97\times 10^{-7}$} & {1.74}  \\
    {$1.5$} & {$1.37\times 10^{-6}$} & {$1.54\times 10^{-6}$} & {1.39}  \\
    \bottomrule
  \end{tabular}
  \end{small}

  \caption{Performance of the {Lasso} regression with different values of
    $\gamma$, where the target circuit is $\mathsf{vqe}$-$6$-$4$.}%
  \label{tab:lambada_mse_cost}
\end{table}

In \cref{fig:neighbor_and_learning_alg}, we also compare the performance of our
approach with neural network-based approaches.
Here, the size of training circuits is set to 5,000.
It turns out that in all cases we have studied, linear regression significantly
outperforms neural networks in QEM performance, indicating that nonlinearity may
not offer advantages in our tasks.
\begin{figure}[H]
  \captionsetup[subfigure]{justification=centering, labelfont = bf}
  \centering
  \begin{subfigure}[t]{0.33\textwidth}
    \centering
    \includegraphics[width = \linewidth]{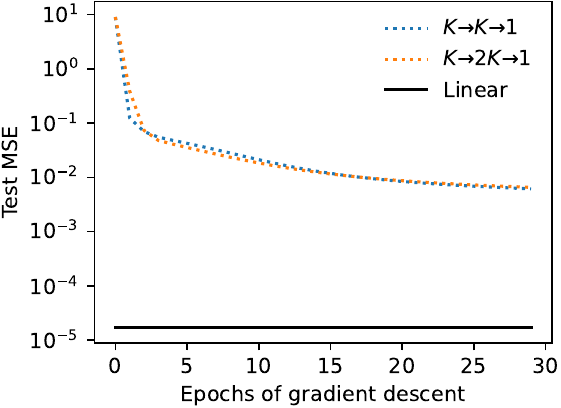}
    \subcaption{CPTP neighbors (weight-1)}
  \end{subfigure}%
  \hfill
  \begin{subfigure}[t]{0.33\textwidth}
    \centering
    \includegraphics[width = \linewidth]{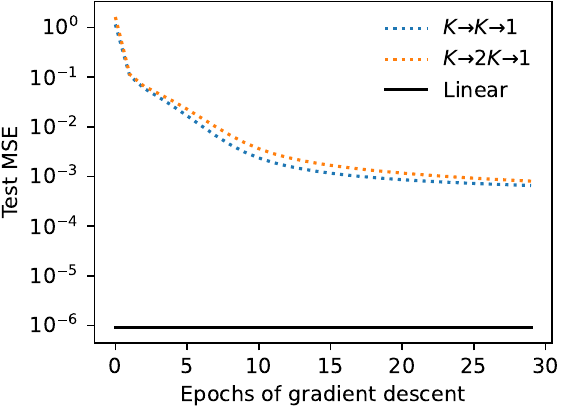}
    \subcaption{Pauli neighbors (weight-1)}
  \end{subfigure}%
  \begin{subfigure}[t]{0.33\textwidth}
    \centering
    \includegraphics[width = \linewidth]{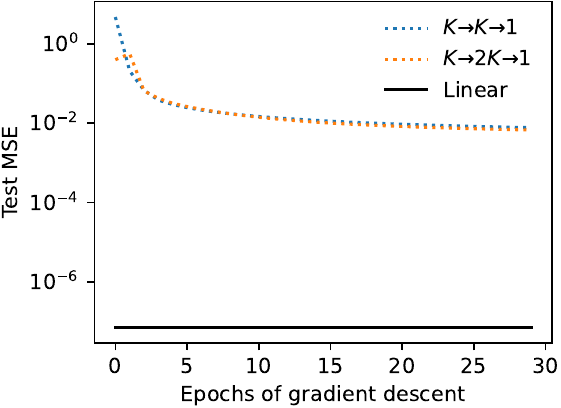}
    \subcaption{Pauli neighbors (weight-3)}
  \end{subfigure}%
  \hfill
  \caption{Comparisons between the performances of learning-based QME models
    with different neighbors and different learning models.
    Here neighbor circuits can be Pauli neighbors or CPTP neighbors, and
    learning models can be linear models or neural networks ($K \to K \to 1$ and
    $K \to 2K \to 1$).
    The target circuit is $\mathsf{vqe}$-$6$-$4$.}%
  \label{fig:neighbor_and_learning_alg}
\end{figure}

\section{Quantum Chemistry Problems}\label{app:quantum_chemistry_problems}

Following the approach in Ref.~\cite{guo_experimental_2024}, we study the LiH
and $\text{F}_2$ molecular Hamiltonians.
The Hamiltonians for both molecules are generated in the STO-3G basis set and
then transformed into qubit Hamiltonians using the Jordan-Wigner transformation.
This process is automatically implemented by the
\texttt{qchem.molecular\_hamiltonian} function in
\texttt{PennyLane}~\cite{bergholm2022pennylane} with the following specific
settings.
For LiH, we set
\begin{itemize}
  \item \texttt{coordinates = [[0.3925, 0.0, 0.0],[-1.1774, 0.0, 0.0]]},
  \item \texttt{active\_electrons = 2},
  \item \texttt{active\_orbitals = 3}.
\end{itemize}
For $\text{F}_2$, we set
\begin{itemize}
    \item \texttt{coordinates = [[0.0, 0.0, -0.7059],[0.0, 0.0, 0.7059]]},
    \item \texttt{active\_electrons = 10},
    \item \texttt{active\_orbitals = 6}.
\end{itemize}
The resulting Hamiltonian is 6-qubit for LiH and 12-qubit for $\text{F}_2$.

The ansatz circuits for VQE in both cases are chosen to be unitary-coupled
clusters, where a series of circuit simplification and compilation strategies
from Ref.~\cite{guo_experimental_2024} are applied.
The ansatz circuit for LiH is shown in \cref{fig:UCC_ansatz} (a) and that for
$\text{F}_2$ is shown in \cref{fig:UCC_ansatz} (b).
The gates before the dotted lines form the quantum circuits for preparing the
initial multi-reference states.
Here, the parameterized gates are represented by the yellow blocks, which are
all Pauli rotation gates.
\begin{figure}[H]
  \centering
  \captionsetup[subfigure]{justification=centering, labelfont = bf}
  \begin{subfigure}[t]{\textwidth}
    \centering
    \includegraphics[width=1\linewidth]{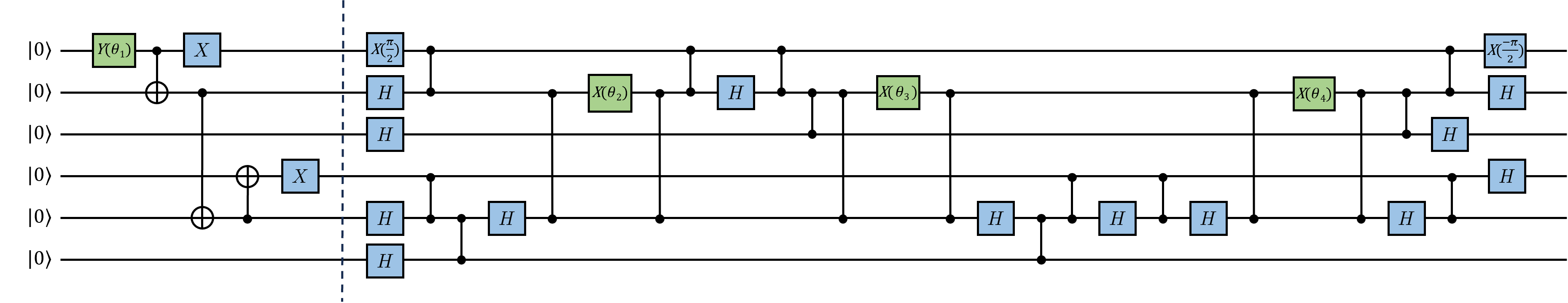}
    \subcaption{LiH}
  \end{subfigure}

  \vspace{\baselineskip}

  \begin{subfigure}[t]{\textwidth}
    \centering
    \includegraphics[width=1\linewidth]{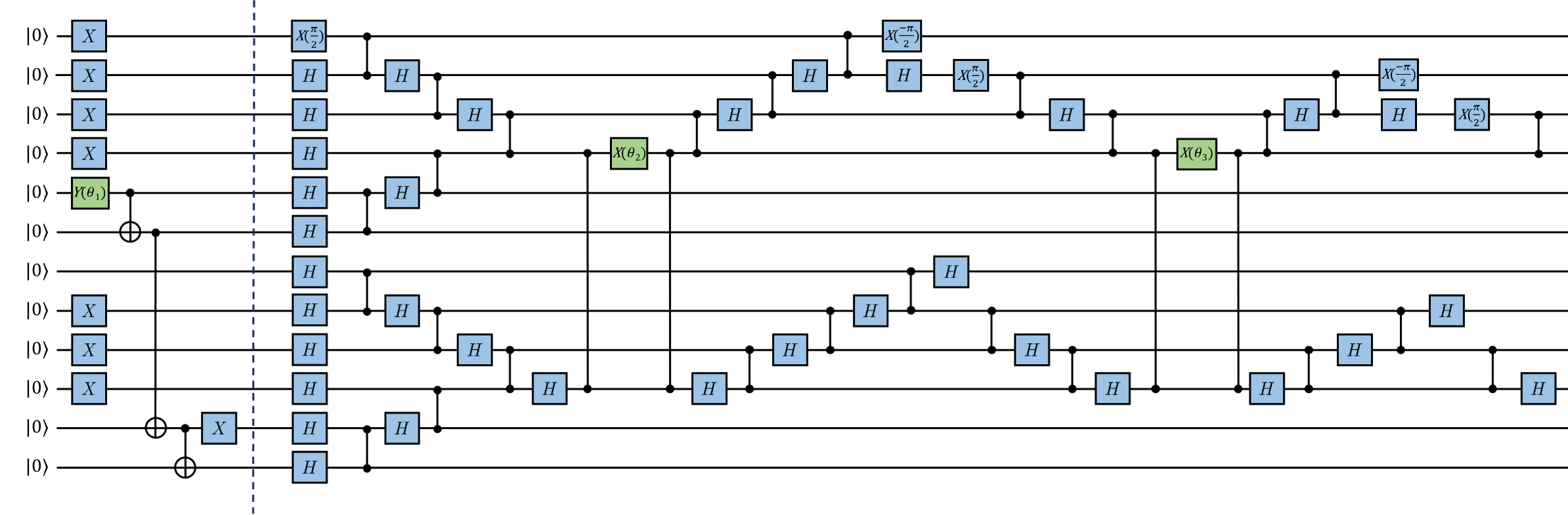}
    \subcaption{$\text{F}_2$ part 1}
  \end{subfigure}

  \begin{subfigure}[t]{\textwidth}
    \centering
    \includegraphics[width=1\linewidth]{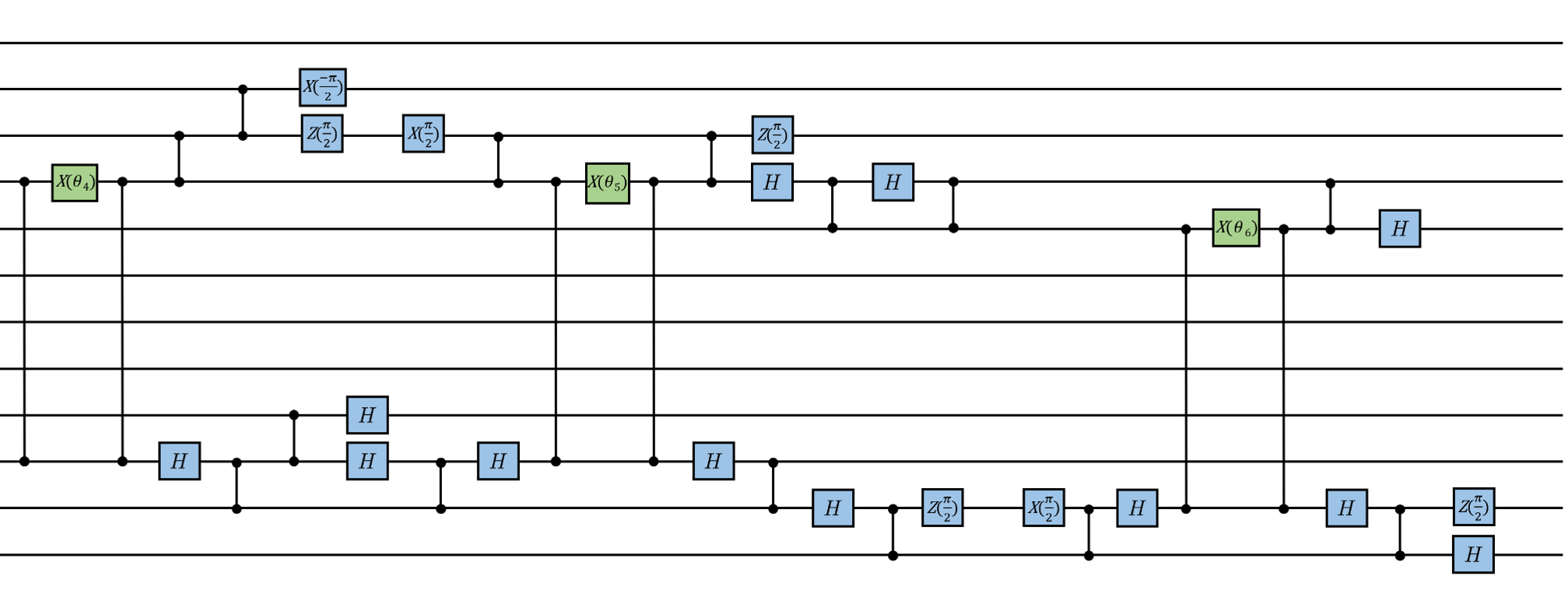}
    \subcaption{$\text{F}_2$ part 2}
  \end{subfigure}

  \caption{The compiled UCC ansatz circuits for LiH and $\text{F}_2$.
    The circuits in subfigure (b) and subfigure (c) together comprise the
    compiled UCC ansatz circuits for $\text{F}_2$.
    The part of the circuit before the blue dashed line represents the initial
    state preparation.
    The gates marked in green, along with all the two-qubit gates, are fixed
    Clifford gates, while the gates marked in yellow are single-parameter
    rotation gates.}\label{fig:UCC_ansatz}
\end{figure}
\subsection{Performance of NIL on UCC Ansatz Circuits for $\text{F}_2$}
In this subsection, we demonstrate the application of our method to noise mitigation in UCC ansatz circuits for $\mathrm{F}_2$, with the circuit diagrams shown in \cref{fig:UCC_ansatz}(b) and \cref{fig:UCC_ansatz}(c).

Specifically, we use only half of the weight-1 Pauli neighbors as the set of
neighbor circuits, and then generate 1,000 training circuits by
the 2-design training method.
To obtain each noisy expectation value for the training circuits, we run each circuit 10,000 times and calculate the expectation value as the output.
The label for each training circuit is the exact expectation value, which can be
obtained efficiently since it is a Clifford circuit.
Finally, we apply Lasso regression to fit the training set.
As shown in \cref{subfig:performance_vqef2}, the training MSE ($\approx$ test
MSE) is nearly $10^{-4}$, which is comparable to the statistical fluctuations
caused by measurement shot noise. Compared with the MSE of the noisy outputs ($\approx 0.25$), the mitigation gain reaches nearly two orders of magnitude, corresponding to a suppression of about 99\% of the errors in terms of MSE.
This numerical result demonstrates that our approach is applicable to real quantum chemistry
tasks of modest size.

\begin{figure}[htbp]
  \centering
  \includegraphics[width = 0.4\columnwidth]{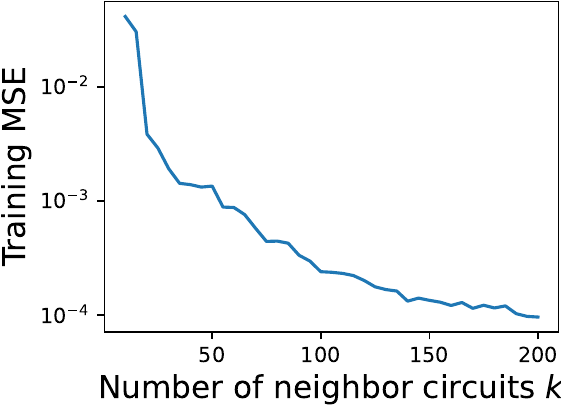}
  \caption{The 12-qubit UCC ansatz circuit for $\text{F}_2$.
    In this case, the MSE of the noisy outputs without mitigation is 0.25.}%
  \label{subfig:performance_vqef2}
\end{figure}%